\documentclass[a4paper]{scrartcl}

\usepackage[utf8]{inputenc}
\usepackage{latexsym}
\usepackage{amsmath}
\usepackage{amsthm}
\usepackage{amssymb}
\usepackage{mathtools}
\usepackage{graphicx}
\usepackage{ifthen}
\usepackage{calc}
\usepackage{stmaryrd}
\usepackage{enumitem}
\usepackage[dvipsnames]{xcolor}
\usepackage[colorlinks,linkcolor=RedViolet,urlcolor=NavyBlue,citecolor=ForestGreen]{hyperref}
\usepackage{alphalph}
\usepackage{BOONDOX-cal}

\usepackage{chemfig}

\usepackage[backend=bibtex,style=numeric,giveninits=true]{biblatex}
\renewbibmacro*{doi+eprint+url}{\iftoggle{bbx:url}
  {\iffieldundef{doi}{\usebibmacro{url+urldate}}{}}
  {}\newunit\newblock
  \iftoggle{bbx:eprint}
  {\usebibmacro{eprint}}
  {}\newunit\newblock
  \iftoggle{bbx:doi}
  {\printfield{doi}}
  {}}

\allowdisplaybreaks

\makeatletter

\newsavebox{\tempbox}
\renewcommand{\@makecaption}[2]{
  \vspace{10pt}
  \sbox{\tempbox}{\textbf{#1.} #2}
  \ifthenelse{\lengthtest{\wd\tempbox > \linewidth}}{
    \textbf{#1.} #2\par
  }{
    \begin{center}
      \textbf{#1.} #2
    \end{center}
  }
}

\makeatother

\numberwithin{equation}{section}

\numberwithin{figure}{section}

\newtheoremstyle{mythm}{}{}{\itshape}{}{\bfseries}{.}{.5em}{\thmname{#1}~\thmnumber{#2}\ifthenelse{\equal{\thmnote{#3}}{}}{}{~(\thmnote{#3})}}

\newtheoremstyle{mydefn}{}{}{\upshape}{}{\bfseries}{.}{.5em}{\thmname{#1}~\thmnumber{#2}\ifthenelse{\equal{\thmnote{#3}}{}}{}{~(\thmnote{#3})}}

\newtheoremstyle{myremark}{}{}{\upshape}{}{\itshape}{.}{.5em}{\thmname{#1}~\thmnumber{#2}\ifthenelse{\equal{\thmnote{#3}}{}}{}{~(\thmnote{#3})}}

\theoremstyle{mythm}
\newtheorem{theorem}{Theorem}[section]
\newtheorem{lemma}[theorem]{Lemma}
\newtheorem{proposition}[theorem]{Proposition}
\newtheorem{corollary}[theorem]{Corollary}
\newtheorem{fact}[theorem]{Fact}

\theoremstyle{mydefn}

\newtheorem{example}[theorem]{Example}

\theoremstyle{myremark}

\theoremstyle{mythm}

\newcommand{\uend}{\hfill$\lrcorner$}
\newcommand{\uende}{\eqno\lrcorner}

\newcounter{claimcounter}

\newlist{caselist}{description}{10}
\setlist[caselist]{font=\itshape\mdseries}

\setenumerate[1]{label=(\arabic*)}
\newlist{eroman}{enumerate}{2}
\setlist[eroman,1]{label=(\roman*)}
\setlist[eroman,2]{label=(\alph*)}
\newlist{ealph}{enumerate}{1}
\setlist[ealph]{label=(\Alph*)}

\newcounter{nlistcounter}

\usepackage{color}
\definecolor{blau}{RGB}{0,84,159}
\definecolor{hellblau}{RGB}{142,168,229}
\definecolor{petrol}{RGB}{0,97,101}
\definecolor{tuerkis}{RGB}{0,152,161}
\definecolor{gruen}{RGB}{87,171,39}
\definecolor{maigruen}{RGB}{189,205,0}
\definecolor{gelb}{RGB}{255,237,0}
\definecolor{orange}{RGB}{255,128,0}
\definecolor{magenta}{RGB}{227,0,102}
\definecolor{rot}{RGB}{204,7,30}
\definecolor{bordeaux}{RGB}{161,16,53}
\definecolor{violett}{RGB}{97,33,88}
\definecolor{lila}{RGB}{122,111,172}
\definecolor{grey}{gray}{0.7}
\definecolor{mittelblau}{RGB}{0,128,255}

\newcommand{\bigmid}{\mathrel{\big|}}
\newcommand{\Bigmid}{\mathrel{\Big|}}
\newcommand{\ceil}[1]{\left\lceil#1\right\rceil}
\newcommand{\floor}[1]{\left\lfloor#1\right\rfloor}
\newcommand{\angles}[1]{\left\langle#1\right\rangle}
\renewcommand{\tilde}{\widetilde}
\renewcommand{\hat}{\widehat}

\renewcommand{\vec}[1]{\boldsymbol{#1}}
\newcommand{\symdiff}{\mathbin{\triangle}}

\DeclareMathOperator{\rg}{rg}

\newcommand{\Fraisse}{Fra\"{\i}ss{\'e}}

\renewcommand{\phi}{\varphi}
\renewcommand{\epsilon}{\varepsilon}

\newcommand{\Nat}{{\mathbb N}}
\newcommand{\PNat}{{\mathbb N}_{>0}}
\newcommand{\Real}{{\mathbb R}}
\newcommand{\PReal}{\Real_{> 0}}
\newcommand{\NNReal}{\Real_{\ge 0}}

\newcommand{\Rat}{{\mathbb Q}}

\newcommand{\LL}{\textsf{\upshape L}}
\newcommand{\LC}{\textsf{\upshape C}}
\newcommand{\FO}{\textsf{\upshape FO}}

\newcommand{\CC}{{\mathcal C}}

\newcommand{\CF}{{\mathcal F}}
\newcommand{\CG}{{\mathcal G}}

\newcommand{\CT}{{\mathcal T}}

\newcommand{\CX}{{\mathcal X}}

\usepackage[textwidth=2cm]{todonotes}

\usepackage{algorithmic}

\algsetup{linenodelimiter=.}

\usepackage{float}
\newfloat{algorithm}{ht}{alg}
\floatname{algorithm}{Algorithm}

\newcommand{\Ex}{\textup{E}}

\DeclareMathOperator{\Inj}{Inj}
\DeclareMathOperator{\Surj}{Surj}
\DeclareMathOperator{\Bij}{Bij}
\DeclareMathOperator{\Sym}{Sym}
\DeclareMathOperator{\Perm}{Perm}
\DeclareMathOperator{\DS}{DS}
\DeclareMathOperator{\Corr}{Corr}
\DeclareMathOperator{\Cp}{Cp}
\DeclareMathOperator{\dist}{dist}

\newcommand{\ged}{\delta_{\textup{ed}}}

\newcommand{\wl}[1]{\vec{wl}(#1)}

\newcommand{\VWL}{\mathbb V_{\textup{WL}}}
\newcommand{\KWL}{K_{\textup{WL}}}
\newcommand{\anglesWL}[1]{\angles{#1}_{\textup{WL}}}
\newcommand{\deltaWL}{\delta_{\textup{WL}}}

\newcommand{\VF}{\mathbb V_{\CF}}
\newcommand{\anglesF}[1]{\angles{#1}_{\CF}}
\newcommand{\deltaF}{\delta_{\CF}}
\newcommand{\dsamp}{\delta_{\textup{samp}}}

\DeclareMathOperator{\hd}{hd}
\DeclareMathOperator{\ed}{ed}
\DeclareMathOperator{\sd}{sd}
\DeclareMathOperator{\emb}{emb}
\DeclareMathOperator{\semb}{semb}

\begin{document}
\title{Some Thoughts on Graph Similarity\thanks{This is an expanded
    version of \cite{Grohe24}}}
\author{Martin Grohe\thanks{Funded by the European Union (ERC, SymSim,
101054974). Views and opinions expressed are however those of the
author(s) only and do not necessarily reflect those of the European
Union or the European Research Council. Neither the European Union
nor the granting authority can be held responsible for them.}}
\maketitle

\begin{abstract}
  We give an overview of different approaches to measuring the
  similarity of, or the distance between, two graphs, highlighting
  connections between these approaches. We also discuss the complexity
  of computing the distances.  
\end{abstract}

\section{Introduction}
Graphs, or networks, are basic models ubiquitous in science and
engineering.
They are used to describe a diverse range of objects and
processes, including, for example, chemical
compounds, social interaction, molecular interaction, and computational processes. To
understand and classify graph models, we need to compare graphs. Since
data and models cannot always guaranteed to be exact, it
is essential to understand what makes two graphs similar or
dissimilar. There are many different approaches to similarity, for
example, based on edit distance, common subgraphs, spectral
similarity, behavioural equivalence, or as a limit case isomorphism.
All of these are relevant in some
situations because different applications
have different demands to similarities.

Applications of graph
similarity occur in many different areas, among
them combinatorial optimisation, computational biology, computer vision, 
database systems, data mining, formal verification, and machine learning.
Traditional applications of isomorphism and similarity can be
described as \emph{graph matching} applications: real-world entities,
ranging from chemical molecules to human faces, are modelled by
graphs, and the goal is to match, or align them. Examples of this type
of application appear in chemical information systems, computer
vision, and database schema matching (see, e.g.,
\cite{confogsanven04,emmdehshi16}). Many of these applications do not
require exact isomorphisms. Instead, a relaxation to similarity is
often sufficient and even desirable for dealing with imprecise and
uncertain data. A second important form of application exploits
similarity within a single graph, or, in the limit case
symmetries of the graph, to
design more efficient algorithms for hard algorithmic problems. We may
use symmetries to prune search trees in backtracking algorithms (e.g.\ in SAT solving \cite{sak09}) or to reduce the size of
instances of hard algorithmic problems (e.g.\ in
optimisation~\cite{bodherjos13,GroheKMS14} or model
checking~\cite{claemejhasis98}).

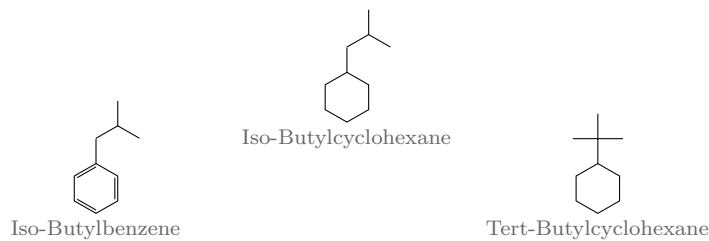
\begin{figure}
  \centering
  \begin{tikzpicture}[scale=1.1]

    \begin{scope}
      \draw
      (90:3mm)--(150:3mm)--(210:3mm)--(270:3mm)--(330:3mm)--(30:3mm)--cycle;
      \draw (90:2.6mm)--(150:2.6mm) (210:2.6mm)--(270:2.6mm)
      (330:2.6mm)--(30:2.6mm);
      \draw (90:3mm) -- (0,6mm);
      \draw[xshift=2.6mm,yshift=7.55mm] (0,0) -- (90:3mm) (0,0) --
      (210:3mm) (0,0) -- (330:3mm);
      \path (0,-0.5) node {\color{black!60}\scriptsize Iso-Butylbenzene}; 
    \end{scope}
    
    \begin{scope}[xshift=60mm]
      \draw 
      (90:3mm)--(150:3mm)--(210:3mm)--(270:3mm)--(330:3mm)--(30:3mm)--cycle; 
      \draw (90:3mm) -- (0,9mm) (-0.3,0.6) -- (0.3,0.6); 
      
      \path (0,-0.5) node {\color{black!60}\scriptsize Tert-Butylcyclohexane}; 
    \end{scope}
    
    \begin{scope}[xshift=30mm,yshift=11mm]
      \draw
      (90:3mm)--(150:3mm)--(210:3mm)--(270:3mm)--(330:3mm)--(30:3mm)--cycle;
      \draw (90:3mm) -- (0,6mm);
      \draw[xshift=2.6mm,yshift=7.55mm] (0,0) -- (90:3mm) (0,0) --
      (210:3mm) (0,0) -- (330:3mm);
      
      \path (0,-0.5) node {\color{black!60}\scriptsize Iso-Butylcyclohexane}; 
      
    \end{scope}
  \end{tikzpicture}
  \caption{Which of these three molecular graphs are most similar?}
  \label{fig:fuel}
\end{figure}

It is not clear what makes two graphs similar. Consider the three
molecular graphs in Figure~\ref{fig:fuel}. They all look somewhat
similar. Usually, we will study these graphs, or the molecules they represent,
in a specific application scenario, say, the design of synthetic
fuels. It turns out that the two bottom molecules are more similar
with respect to their relevant chemical properties. Can we design a
graph similarity measure in a way that it puts these two graphs closer
together than either of them to the third? We will not answer this
question here (in practice, we might try to learn such a similarity
measure from data), but the example illustrates that similarity may
very much depend on specific applications, and there is certainly not
a unique similarity measure suitable for all, or most,
applications.

The purpose of this paper is to sketch a theory of graph
similarity. We will study different ``principled'' approaches to graph
similarity, highlighting connections between these approaches. Central
to my thinking about similarity are two different views on what makes
two graphs similar. Under the \emph{operational view}, two graphs are
similar if one can easily be transformed into the other. Under the
\emph{declarative view}, two graphs are similar if they have similar
properties. An example of an operational similarity is \emph{edit
  distance}, measuring how many edges must be added and deleted from
one graph to obtain a graph isomorphic to the other. \emph{Graph
  kernels} used in machine learning provide examples for the
declarative view. In a nutshell, graph kernels collect numeric
features, for example, the number of triangles, and provide an inner
product on the feature space. The operational view is algorithmic in
nature, whereas the declarative view is rooted in logic and
semantics.  Ultimately, it will be a goal of our theory to establish
connections between the approaches, viewing them as dual in some sense.

Both views on similarity are also important from a practical
perspective. We need to have algorithms measuring similarity, and
often we not only need to measure how similar two graphs are but we
actually want to transform one to the other (think of an application
like database repairs). But we also need to have a semantical
understanding of what we are doing, that is, we want to explain
which features make two graphs similar.

\section{Preliminaries}

$\Nat$ denotes the set of natural numbers (nonnegative integers),
$\PNat$ the set of positive integers, $\Real$ the set of reals, and
$\NNReal$ the set of nonnegative reals. For every $n\in\Nat$, we let
$[n]\coloneqq\{1,\ldots,n\}$.

We
denote the power set of a set $V$ by $2^V$ and the set of all $k$-element
subsets of $V$ by $\binom{V}{k}$. For sets $V,W$, we denote the set of
all mappings from $V$ to $W$ by $W^V$, the set of all injective
mappings from $V$ to $W$ by $\Inj(V,W)$, the set of all surjective
mappings from $V$ to $W$ by $\Surj(V,W)$ and the set of all bijective
mappings from $V$ to $W$ by $\Bij(V,W)$. Moreover, we denote the set
of all
permutations of $V$ by $\Sym(V)$ (so $\Sym(V)=\Bij(V,V)$).

We assume an understanding of basic probability. We will use
Hoeffding's well-known concentration inequality in the following form
(see \cite[Theorem~4.12]{MitzenmacherU17}).

  \begin{fact}[Hoeffding's Inequality]\label{fact:hoeffding}
    Let $a,b\in\Real$ with $a<b$. Let $X_1,\ldots,X_n$ be a sequence of independent identically distributed
    random
    variables with $\rg(X_i)\subseteq[a,b]$, and let $X{\coloneqq}\sum_{i=1}^nX_i$. Then
    for all $\epsilon\in\PReal$:
    \[
      \Pr\Big(\big|X-\Ex(X)\big|\ge \epsilon n\Big)\le
      2\exp\left(-\frac{2\epsilon^2}{(b-a)^2}n\right).
    \]
  \end{fact}

\subsection{Vectors and Matrices}
It will be convenient for us to index vectors and matrices by arbitrary
sets. Formally, a vector $\vec a\in X^I$ over some set $X$ is just a
mapping from $I$ to $X$, and its entries are $\vec a(i)$ for $i\in I$.
A matrix $A\in X^{I\times J}$ over 
$X$ is a
mapping from $I\times J$ to $X$. The vectors $A(i,\cdot)\in X^J$ are
the \emph{rows} of the matrix and the vectors $A(\cdot,j)\in X^I$ the
\emph{columns}. A \emph{permutation matrix} is a
matrix over $\{0,1\}$ that has exactly one $1$-entry in
every row and in every column. A permutation matrix
$P\in\{0,1\}^{I\times J}$ corresponds to a bijective mapping $\pi:I\to
J$ defined by letting $\pi(i)$ be the unique $j$ such that $P(i,j)=1$.

For a matrix $A\in \Real^{I\times I'}$ and bijections $\pi:I\to
J$, $\pi':I'\to J'$, we let
$A^{\pi,\pi'}\in X^{J\times J'}$ be the matrix with entries
$A^{\pi,\pi'}(j,j')\coloneqq A(\pi^{-1}(j),(\pi')^{-1}(j'))$. Then if $\pi,\pi'$ are the
permutations corresponding to the permutation matrices
$P\in \{0,1\}^{I\times J}$, $P'\in\{0,1\}^{I'\times J'}$, we have
\[
  A^{\pi,\pi'}=P^\top AP'.
\]
For $A\in\Real^{I\times I}$ and $\pi\in\Bij(I,J)$ we let $A^{\pi}\coloneqq A^{\pi,\pi}$.

\subsection{Graphs}
We use a standard graph notation based on \cite{Diestel16}. Graphs are
always finite and, unless explicitly stated otherwise, undirected and
simple. We denote the vertex set of graph $G$ by $V_G$, or just $V$ if
the graph is clear from the context, and similarly, we denote the edge
set by $E_G$ or just $E$. The
\emph{order} $|G|$ of a graph $G$ is the number of vertices,
that is, $|G|\coloneqq|V_G|$.  We denote
edges of a graph by $vw$ instead of $\{v,w\}$, stipulating
that $vw$ and $wv$ denote the same edge. The set of \emph{neighbours}
of a vertex $v$ is $N_G(v)\coloneqq\{w\mid vw\in E_G\}$, and the
\emph{degree} of $v$ is $\deg_G(v)\coloneqq|N_G(v)|$.
We denote the class of all graphs by
$\CG$. Moreover, for every $n\in\Nat$, we denote the class of all
graphs $G$ with vertex set $V_G=[n]$ by $\CG_n$, and we let $\CG_{\le n}\coloneqq\bigcup_{i=1}^n\CG_i$. 

We write
$G\cong H$ to denote that two graphs $G,H$ are isomorphic,
that is, there is a bijective mapping $\pi: V_G\to V_H$ such that
$vv'\in E_G\iff \pi(v)\pi(v')\in E_H$ for all $v,v'\in V_G$. The
\emph{adjacency matrix} of a graph $G=(V,E)$ is the matrix
$A_G\in\{0,1\}^{V\times V}$ with entries $A_G(v,v')=1$ if $vv'\in E$ and
$A_G(v,v')=0$ otherwise.
Observe that graphs
$G,H$ are isomorphic if and only if there is a bijection
$\pi:V_G\to V_H$ such that $A_G^\pi=A_H$.

\subsection{Metrics}
A \emph{pseudo-metric} on a set $X$ is a mapping $\delta:X\times
X\to\PReal$ such that
\begin{eroman}
\item $\delta(x,x)=0$ for all $x\in X$;
\item $\delta(x,y)=\delta(y,x)$ for all $x,y\in X$;
\item $\delta(x,y)\le\delta(x,z)+\delta(z,y)$ for all $x,y,z\in X$.
\end{eroman}
A pseudo-metric $\delta$ is a \emph{metric} if $\delta(x,y)>0$ for all
$x\neq y$.

We will mainly be interested in pseudo-metrics on graphs. One way to
obtain such metrics is by applying metrics for matrices to the adjacency
matrices of graphs. The most common matrix metrics are
derived from matrix norms. A \emph{pseudo-norm} on a real vector space
$\mathbb V$ is a mapping $\|\cdot\|:\mathbb V\to\PReal$ such that
$\|\vec 0\|=0$, $\|a\vec a\|=|a|\cdot\|\vec a\|$ and
$\|\vec a+\vec b\|\le\|\vec a\|+\|\vec b\|$ for all $\vec a,\vec
b\in\mathbb V$ and $a\in\Real$. A pseudo-norm $\|\cdot\|$ is a
\emph{norm} if $\|\vec a\|>0$ for all $\vec a\neq\vec0$.
From a (pseudo-)norm $\|\cdot\|$ on $\mathbb V$
we directly obtain a (pseudo-)metric $\delta$ on $\mathbb V$ by
letting
\[
  \delta(\vec a,\vec b)\coloneqq \|\vec a-\vec b\|.
\]
The most important norms on
a finite-dimensional real vector space $\Real^I$ are the $\ell_p$ norms
defined by $\|\vec a\|_p\coloneq\left(\sum_{i\in I} |\vec
  a(i)|^p\right)^{1/p}$, for $1\le p\le\infty$. In the limit $p=\infty$, this yields $\|\vec a\|_\infty\coloneq\max_{i\in I} |\vec
  a(i)|$. We call a matrix norm $\|\cdot\|$ \emph{permutation invariant} if for
all matrices $A\in\Real^{I\times J}$ and all permutations
$\pi\in\Sym(I),\rho\in\Sym(J)$ it holds that
$\|A^{\pi,\rho}\|=\|A\|$. All matrix norms we
will consider in this article are permutation invariant. 

The most important matrix norms are defined as
follows, for a matrix $A\in\Real^{I\times J}$.
\begin{itemize}
\item \emph{Entrywise $p$-norm.} 
  \[
    \|A\|_{(p)}\coloneqq\left(\sum_{i\in I,j\in J}
      |A(i,j)|^p\right)^{1/p}
  \]
  for $1\le p<\infty$, and $\|A\|_{(\infty)}\coloneqq \max_{i\in I,j\in J}
      |A(i,j)|$ for $p=\infty$.
  The best-known of these norms is the
  \emph{Frobenius norm} $\|A\|_F\coloneqq\|A\|_{(2)}$.
\item \emph{Operator $p$-norm.}
  \[
    \|A\|_p\coloneqq\sup_{\vec a\in\Real^J\setminus\{\vec 0\}}\frac{\|A\vec
      a\|_p}{\|\vec a\|_p}.
  \]
  The best-known of these norms is the
  \emph{spectral norm} $\|A\|_2$.
\item \emph{Cut norm.}
  \[
    \|A\|_\square\coloneqq\max_{S\subseteq I,T\subseteq
      J}\left|\sum_{i\in S,j\in T}A(i,j)\right|.
  \]
\end{itemize}
Note that
\begin{align}
    \label{eq:1}
    \|A\|_{(1)}=\sum_{i\in I,j\in J} |A(i,j)|&=\max_{S\subseteq
                                               I,T\subseteq
                                               J}\sum_{i\in S,j\in
                                               T}|A(i,j)|\\
  \notag
  &\ge
  \max_{S\subseteq I,T\subseteq J}\bigg|\sum_{i\in S,j\in
      T}A(i,j)\bigg|=\|A\|_\square.
\end{align}
It is also worth noting that
\[
  \|A\|_\square\le\sup_{\vec a\in\Real^J\setminus\{\vec
    0\}}\frac{\|A\vec a\|_1}{\|\vec a\|_\infty}\le 4\|A\|_\square,
\]
as observed by Frieze and Kannan~\cite{FriezeK99}.

\subsection{Similarity}

As opposed to ``distance'', formalised by the notion of a metric, there is no universally agreed-upon definition
of ``similarity''. Intuitively, similarity means proximity with respect to some
metric, that is, some kind of inverse of a metric. A standard way of
turning a (pseudo-)metric $\delta$ on a set $X$ into a \emph{similarity
  measure} $\sigma$ on $X$ is by letting
\begin{equation}\label{eq:sim}
  \sigma(x,y)\coloneqq\exp\big(-c\cdot\delta(x,y)\big)
\end{equation}
for some constant $c>0$. Observe that $0<\sigma(x,y)\le1$, with $\sigma(x,x)=1$
(and $\sigma(x,y)<1$ for all $x\neq y$ if $\delta$ is a metric and not just a
pseudo-metric) and $\sigma(x,y)\to0$ as $\delta(x,y)\to\infty$. The
parameter $c$ controls how fast this convergence is.

For a normalised metric $\delta$ with range in $[0,1]$, it may be
easiest to define a similarity measure $\sigma$ by simply letting
$\sigma(x,y)\coloneqq 1-\delta(x,y)$.

For the rest of this paper, we work with metrics even though
intuitively, we are interested in similarity. We keep in mind that
\eqref{eq:sim} gives us a generic way of turning a metric into a
similarity measure.

\subsection{Graph Metrics}

A \emph{graph metric} is a pseudo-metric
$\delta:\CG\times\CG\to\NNReal$ on the class $\CG$ of all
graphs that is \emph{isomorphism invariant}, which means that
$\delta(G,H)=\delta(G',H')$ whenever $G\cong G',H\cong H'$. Note that a
graph metric $\delta$ always satisfies $\delta(G,H)=0$ if $G$ and $H$
are isomorphic. Also note that we do
not require $\delta(G,H)\neq 0$ if $G\not\cong H$. So graph metrics
are always pseudo-metrics even if we factor out
isomorphism.\footnote{We could make a distinction between \emph{graph
    metrics} $\delta$ satisfying $\delta(G,H)=0\Leftrightarrow G\cong
  H$ and \emph{graph pseudo-metrics} only satisfying $\delta(G,H)=0\Leftarrow G\cong
  H$. However, this distinction is not important in this article.}

Note that, trivially, we can define a graph metric $\delta_{\cong}$ by
\begin{equation}
  \label{eq:2}
    \delta_{\cong}(G,H)\coloneqq
  \begin{cases}
    0&\text{if }G\cong H,\\
    1&\text{otherwise}.
  \end{cases}
\end{equation}
We call $\delta_{\cong}$ the \emph{isomorphism distance}.

\section{Operational Distances}
Recall that ``operational distances'' on graphs are based on the
intuition that two graphs should be similar, or close with respect to
the metric, if they can easily be transformed into each
other. \emph{Graph edit distance}, to be discussed in
Section~\ref{sec:ged}, is the prototypical operational distance.

It will be convenient to restrict the operational metrics to graphs
of the same order first. It is not entirely obvious how to generalise
the definitions to graphs of different orders; we will discuss this in
Section~\ref{sec:sizes}. \emph{Throughout this section, we will
  consider graphs $G=(V,E_G)$ and
  $H=(W,E_H)$. In Sections \ref{sec:ged} and \ref{sec:matrix}, these
  graphs will have the same order $n\coloneqq|V|=|W|$. After that,
  they will have orders $m\coloneqq|V|,n\coloneqq|W|$.}

\subsection{Graph Edit Distance}
\label{sec:ged}
We define the \emph{edit distance} $\ged(G,H)$ between $G$ and $H$ to
be the minimum number of edges to be added or removed from $G$ to
obtain a graph isomorphic to $H$, that is:
\begin{equation}
  \label{eq:3}
    \ged(G,H)\coloneqq\min\Big\{ |D|\Bigmid D\subseteq{\textstyle\binom{V}{2}}\text{
    such that }\big(V,E_G\symdiff D\big)\cong H\Big\}. 
\end{equation}
Here $\symdiff$ denotes symmetric difference. It is easy to verify
that $\ged$ is indeed a graph metric.

We can express the graph edit distance in terms of common
subgraphs. We call a graph $F$ a
\emph{common subgraph} of $G$ and $H$ if there is a subgraph
$F'\subseteq G$ and a subgraph $F''\subseteq H$ such that $F\cong
F'\cong F''$. We call $F$ a \emph{maximum common subgraph} if $|E_F|$
is maximum. Observe that if $F$ is a maximum common subgraph of $G$ and $H$ we have
\[
  \ged(G,H)=|E_G|+|E_H|-2|E_F|.
\]
A third way of expressing graph edit distance is in terms of the
adjacency matrices of the graphs: it is easy to see that 
\begin{equation}
  \label{eq:4}
  \ged(G,H)=\frac{1}{2}\min_{\pi\in\Bij(V,W)}\big\|A_G^\pi-A_H\big\|_{(1)}=\frac{1}{2}\min_{\pi\in\Bij(V,W)}\big\|A_G^\pi-A_H\big\|_{F}^2 .
\end{equation}
The second equality is based on the simple observation that for all
matrices over $\{-1,0,1\}$ the entrywise $1$-norm coincides with
the squared Frobenius norm. We need the factor $\frac{1}{2}$ because
every edge appears twice in the matrix, once in every direction.

\subsection{Graph Distances based on Matrix Norms}
\label{sec:matrix}
The characterisation \eqref{eq:4} of $\ged$ in terms of the entrywise 1-norm or
the Frobenius norm prompts the question of what happens if we take other
matrix norms to define graph metrics. Such graph metrics have been
studied in \cite{GervensG22}. For every permutation-invariant
matrix norm $\|\cdot\|$ we can define a graph metric
$\delta_{\|\cdot\|}$ by
\[
  \delta_{\|\cdot\|}(G,H)=\min_{\pi\in\Bij(V,W)}\big\|A_G^\pi-A_H\big\|.
\]
Note that we need the matrix norm to be permutation invariant, because
otherwise $\delta_{\|\cdot\|}$ might not be symmetric. 

We think of the bijection $\pi:V\to W$ as an \emph{alignment} between
the two graphs. It transforms the graph $G$ into a graph that is as
similar as possible to $H$ with respect to the metric induced by the
matrix norm $\|\cdot\|$. From this perspective it is clear why we
consider such graph metrics based on matrix norms as ``operational''.

We denote
$\delta_{\|\cdot\|}$ for $\|\cdot\|$ being $\|\cdot\|_{(p)}$,
$\|\cdot\|_{p}$, $\|\cdot\|_{\square}$ by $\delta_{(p)}$, $\delta_p$,
$\delta_\square$, respectively.

We have already observed in \eqref{eq:4} that $\delta_{(1)}=2\ged$,
that is, $\delta_{(1)}$ is essentially the
graph edit distance $\ged$. In the following, we prefer to work with
$\delta_{(1)}$ instead of $\ged$, but (sloppily) refer to
$\delta_{(1)}$ as ``graph edit distance'' to convey the intuition.
Also note that $\delta_{(\infty)}$ is the
isomorphism distance $\delta_{\cong}$ introduced in \eqref{eq:2}. It
was shown in \cite{GervensG22} that the graph metrics $\delta_1$,
$\delta_\infty$ coincide and have an
interesting interpretation as a ``local edit distance''. To make this
precise, for a set 
$D\subseteq\binom{V}{2}$ we define $\deg(D)$ to be the maximum degree of the graph
$(V,D)$.

\begin{proposition}[Gervens and Grohe~\cite{GervensG22}]
  \label{prop:local-ed}
  \begin{equation}
    \label{eq:5}
    \delta_1(G,H)=\delta_\infty(G,H)=\min\Big\{\deg(D)\Bigmid D\subseteq{\textstyle\binom{V}{2}}\text{
    such that }\big(V,E_G\symdiff D\big)\cong H\Big\}. 
\end{equation}
\end{proposition}

This explains why we think of $\delta_1$ (and, equivalently,
$\delta_\infty$) as a ``local'' edit distance: compare
the right-hand side of \eqref{eq:5} with \eqref{eq:3}, the definition
of graph edit distance. Instead of the total number of edges in $D$
here we count the maximum number of edges incident with a single vertex.

To prove Proposition~\ref{prop:local-ed}, we use the following well-known
fact.

\begin{fact}\label{fact:operator}
  For every matrix $B\in\Real^{I\times J}$ we have
\begin{align}
  \label{eq:10}
  \|B\|_1&=\max_{j\in J}\sum_{i\in I}|B(i,j)|,\\
  \label{eq:11}
  \|B\|_\infty&=\max_{i\in I}\sum_{j\in J}|B(i,j)|.
\end{align}
\end{fact}

\begin{proof}[Proof of Proposition~\ref{prop:local-ed}]
Since the matrix $A^\pi_G-A_H$ is symmetric for all $\pi\in\Bij(V,W)$,
\eqref{eq:10} and \eqref{eq:11} imply
$\delta_1(G,H)=\delta_\infty(G,H)$. Thus to prove \eqref{eq:5}, it
suffices to prove that $\delta_\infty(G,H)$ equals
\[
  M\coloneqq \min\Big\{\deg(D)\Bigmid D\subseteq{\textstyle\binom{V}{2}}\text{
    such that }\big(V,E_G\symdiff D\big)\cong H\Big\}.
\]
To prove $\delta_\infty(G,H)\le M$, let
$D\subseteq\binom{V}{2}$ such that $\big(V,E_G\symdiff D\big)\cong
H$, and let $\pi:V\to W$ be an isomorphism from $\big(V,E_G\symdiff
D\big)$ to $H$. Let $B\coloneqq A_G^\pi-A_H$. Then $B(w,w')=1$ if $\pi^{-1}(w)\pi^{-1}(w')\in D\cap
E_G$ and $B(w,w')=-1$ if $\pi^{-1}(w)\pi^{-1}(w')\in D\cap
\pi^{-1}(E_H)$ and $B(w,w')=0$ otherwise. Let $v^*\in V$ such that
\[
  \sum_{w\in W}|B(\pi(v^*),w)|= \max_{v\in V}\sum_{w\in W}|B(\pi(v),w)|=
  \max_{v\in V}\deg_{(V,D)}(v)=\deg(D).
\]
Then
\[
  \delta_\infty(G,H)=\min_{\pi\in\Bij(V,W)}\|A_G^{\pi}-A_H\|_\infty\le\|B\|_\infty=
  \sum_{w\in W}|B(\pi(v^*),w)|=\deg(D).
\]
This proves $\delta_\infty(G,H)\le M$.

To prove the converse inequality $\delta_\infty(G,H)\ge M$, let
$\pi\in\Bij(V,W)$ such that
$\delta_\infty(G,H)=\|A^\pi_G-A_H\|_\infty$. Let $D$ be the symmetric difference of
$E_G$ and $\pi^{-1}(E_H)$. We have
\[
  \delta_\infty(G,H)=\|B\|_\infty=\max_{v\in V}\sum_{w\in W}|B(\pi(v),w)|=
  \max_{v\in V}\deg_{(V,D)}(v)=\deg(D).
\]
Since $\pi$ is an isomorphism $(V,E_G\symdiff D)$ to $H$, we have $\deg(D)\ge M$. This proves
the desired inequality and hence completes the proof of \eqref{eq:5}.
\end{proof}

Another very interesting graph metric is
$\delta_\square$.\footnote{This metric is closely related to the
  \emph{cut metric} introduced by Lovász (see \cite{Lovasz12}), but it
  is not the same. We will introduce the actual cut metric in
  Section~\ref{sec:sizes} and denote it by
  $\delta^\odot_\square$. Lovász denotes a normalised version of our $\delta_\square$ by
  $\hat\delta_\square$ (we also use this notation, see
  Section~\ref{sec:normalisation}), and and he denotes $\delta^{\odot}_\square$ by $\delta_{\square}$.}
To explain this metric in purely graph-theoretic terms, for
$X,X'\subseteq V$ we let $e_G(X,X')$ be the number of edges of graph
$G$ with one
endvertex in $X$ and the other endvertex in $X'$, counting edges with
both endvertices in $X\cap X'$ twice. Note
that $e_G(X,X')=\sum_{v\in X,v'\in X'}A_G(v,w)$. A straightforward
calculation shows that \[
  \delta_\square(G,H)=\min_{\pi\in\Bij(G,H)}\max_{X,X'\subseteq
    V}\Big|e_G(X,X')-e_H\big(\pi(X),\pi(X')\big)\Big|.
\]
It follows from \eqref{eq:1} that $\delta_\square$ is bounded from
above by $\delta_{(1)}$: for all graphs $G,H$ of the same order
we have
\begin{equation}
  \delta_\square(G,H)\le \delta_{(1)}(G,H).
\end{equation}
However, $\delta_\square$ may be significantly smaller than
$\delta_{(1)}$. Let us argue that in an important way, $\delta_\square$
captures our intuition of graph similarity much better than
$\delta_{(1)}$. If we consider two graphs chosen uniformly at random from
the class of all $n$-vertex graphs, then for large $n$, we expect them
to look quite similar with high probability---they might differ in
many details, but they will be similar for most relevant global graph
properties. Yet with high probability their edit distance will be
large, of order $\Omega(n^2)$. It turns out that their cut distance is
much lower. Let us sketch a proof that for every $\epsilon>0$ with
probability at least $1-\epsilon$, for two random graphs
$G,H\sim\CG(n,\frac{1}{2})$ we have
$\delta_\square(G,H)=O(n^{3/2})$. Here $\CG(n,\frac{1}{2})$ denotes
the probability distribution of graphs with vertex set $[n]$ where the
edges are drawn independently with probability $\frac{1}{2}$, which is
precisely the uniform distribution. So here our graphs $G,H$ have the
same vertex set $V=W=[n]$. We shall prove that with high probability
it holds that $\|A_G-A_H\|_\square=O(n^{3/2})$. Let us fix sets $X,X'\subseteq[n]$ of
cardinalities $|X|=k,|X'|=k'$. For simplicity, let us assume that $X$
and $X'$ are disjoint and that $k,k'=\Omega(n)$. For
each pair $v\in X,v'\in X'$, we let $Z_{vv'}$ be the random variable that
takes value $1$ if $vv'\in E_G\setminus E_H$, value $-1$ if
$vv'\in E_H\setminus E_G$, and value $0$ otherwise. Then the $Z_{vv'}$
are independent with $\Pr(Z_{vv'}=1)=\Pr(Z_{vv'}=-1)=\frac{1}{4}$ and
$\Pr(Z_{vv'}=0)=\frac{1}{2}$. Moreover, $e_G(X,X')-e_H(X,X')=Z\coloneqq\sum_{v\in
  X,v'\in X'}Z_{vv'}$. The expected value of $Z$ is $0$, and by Hoeffding's Inequality for
all $r\in\PReal$ we have
\[
  \Pr(|Z|\ge rkk')\le e^{-r^2 kk'}.
\]
Then for every constant $c>0$, using that $k,k'=\Omega(n)$ we
can find constants $d,d'>0$ such that with $r=d'n^{-1/2}$ we have
$
  \Pr(|Z|\ge dn^{3/2})\le 2e^{-cn}.
$
So by applying a Union Bound, we bring the probability that for
all $X,X'$ it holds that $|e_G(X,X')-e_H(X,X')|< dn^{3/2}$ arbitrarily close
to $1$.

The difference between $\delta_{(1)}(G,H)=\Omega(n^2)$ and
$\delta_\square(G,H)=O(n^{3/2})$ may not seem that
significant. However, if we normalise the metrics
by multiplying them by $1/n^2$, then this difference becomes more
striking. We will dicuss normalisation in
Section~\ref{sec:normalisation}.

However, before we do this, let me point out that instead of the adjacency matrix, we can also use other matrices to
define graph metrics, for example, the \emph{Laplacian matrix} $L_G$
of a graph or the \emph{distance matrix} $D_G$. Let us only consider
one interesting example based on the distance matrix. The entry $D_G(v,v')$
in the distance matrix is the length of the shortest path from $v$ to
$v'$. Essentially, the distance matrix is the natural representation of
a graph as a finite metric space. Then for $\pi\in\Bij(V,W)$, the
value $\big|D_G(v,v')-D_H\big(\pi(v),\pi(v')\big)\big|$ indicates how
much the distance between $v$ and $v'$ is \emph{distorted} by $\pi$. We
define the \emph{distortion distance} by
\[
  \delta_{\textup{dist}}(G,H)\coloneqq\min_{\pi\in\Bij(V,W)}\|D_G^\pi-D_H\|_{(\infty)}.
\]
The distortion distance is closely related to the
\emph{Gromov-Hausdorff distance} of the two metric spaces represented
by $G$ and $H$ (see Section~\ref{sec:relax}). Compared to the edit
distance or cut distance, the
distortion distance can be more meaningful for sparse graphs because it
is possible to obtain sparse approximations (so-called sparsifiers) of dense graphs 
with respect to the distortion distance (see, for example, \cite{Chew89,FungHHP19,SpielmanT11}).

\subsection{Normalisation}
\label{sec:normalisation}

If we want to compare graphs of different sizes, or just compare
distances across graph sizes, we need to normalise the metrics.

\begin{example}
  Let us try to understand how the edit distance between the complete graph
  $K_{2n}$ and the complete bipartite graph $K_{n,n}$ evolves as $n$ goes
  to infinity. 
  It is easy to see that
  \[
    \ged(K_{2n},K_{n,n})=2\binom{n}{2}=n(n-1),
  \]
  or equivalently, $\delta_{(1)}(K_{2n},K_{n,n})=2n(n-1)$. Thus
  $\lim_{n\to\infty}\ged(K_{2n},K_{n,n})=\infty$, which does not
  convey much insight.

  However, if we scale distances relative to the maximum number of
  possible edges, we obtain the more interesting
  \[
    \frac{\ged(K_{2n},K_{n,n})}{\binom{2n}{2}}=\frac{2
      n(n-1)}{2n(2n-1)}=\frac{n-1}{2n-1}\quad\xrightarrow[n\to\infty]{}\quad\frac{1}{2}.
  \]
 Note that if we consider $\delta_{(1)}$ and
  normalise by $2n(2n-1)$, or even by the simpler
  $(2n)^2$, we get the same limit $\frac{1}{2}$.\uend
\end{example}

The question is how to determine
the right normalisation factor. As a rule of thumb, for many metrics
it seems natural to choose a normalisation factor about the
distance between the complete graph $K_n$ and the empty graph on $n$
vertices. At the same time, we want to keep the normalisation simple,
preferring, for example, $n^2$ over $n(n-1)$.

A natural normalisation factor for edit distance $\delta_{(1)}$ and cut
distance $\delta_\square$ is $n^2$. We define the following normalised
versions of the metrics: for all graphs $G,H$ of order $n$ we let
\begin{align*}
  \widehat\delta_{(1)}(G,H)&\coloneqq\frac{\delta_{(1)}(G,H)}{n^2},\\
  \widehat\delta_{\square}(G,H)&\coloneqq\frac{\delta_{\square}(G,H)}{n^2}.\\
  \intertext{For local edit distance
                                 $\delta_1$, a normalisation factor of
  $n$ seems more natural, and we let}
  \widehat\delta_{1}(G,H)&\coloneqq\frac{\delta_{1}(G,H)}{n}.
\end{align*}
We could continue with other metrics, but these three are the main
metrics we shall study in the rest of this section, so we leave it there.

Let me point out that normalisation can be problematic. In particular,
if we want to study sparse graphs with only $O(n)$ edges, a
normalisation factor $n^2$ for edit distance and cut distance is
unsuitable, because all sparse graphs will end up being very close
to the empty graph. The theory presented here is mainly meaningful for
dense graphs. There have been various attempts to address this issue
(e.g.~\cite{BackhausS20,bolrio11,BorgsCCZ19}), but they are
complicated and beyond the scope of this article.

\subsection{Comparing Graphs of Different Sizes}
\label{sec:sizes}

So far, all metrics that we defined are on graphs of the same
order. So, actually every ``graph metric'' $\delta$ we defined so
far is a family $(\delta_n)_{n\in\Nat}$ where $\delta_n$ is
a (pseudo-)metric on  $\CG_n$. In this section, we want to extend our
metrics in such a way that they also allow it to compare graphs of
distinct orders.

When thinking of graph edit distance, the most obvious way to compare
graphs of distinct orders would be to add another edit operation:
delete or add an isolated vertex. Let us call the version of edit
distance where we allow this operation by $\ged^\oplus$.

We define the following \emph{padding} operation on graphs. For a graph $G$ and a nonnegative
integer $k\in\Nat$, by $G^{\oplus k}$ we denote the graph obtained from $G$
by adding $k$ isolated vertices.\footnote{To be definite, say we add
  the first $k$ nonnegative integers $i\in\Nat$ not in the vertex set of $G$.}
Then it is easy to see that for all graphs $G,H$ we have
\[
  \ged^\oplus (G,H)=
  \begin{cases}
    k+\ged(G^{\oplus k},H)&\text{if }|H|-|G|=k\ge 0,\\
    k+\ged(G,H^{\oplus k})&\text{if }|G|-|H|=k>0.
  \end{cases}
\]
As a variant, we can put different weights $\alpha,\beta$ on the two edit
operations. Then we obtain terms $\alpha k+\beta\ged(\cdots)$. Since
we may often want to ignore isolated vertices, it might be most
natural to let $\alpha=0$ and $\beta=1$.

Similarly, for each of our graph metrics $\delta$ based on matrix
norms we can define a version $\delta^\oplus $ where we pad the smaller of
the two graphs by isolated vertices. These definitions are reasonable
because all the matrix norms we considered are invariant under padding
a matrix by $0$-rows and $0$-columns.

However, this ``additive'' approach to comparing graphs of distinct orders
is often not what we want. In particular, it does not allow us
to approximate a large graph (with many edges) by a small graph,
because the impact of a single edge on the distance is absolute and not
scaled relative to the size.

Lovász~\cite{Lovasz12} proposes a ``multiplicative'' approach based
on the following \emph{blow-up} operation instead of the additive padding
operation. For a graph $G$ and a
nonnegative integer $k\in\Nat$, by $G^{\odot k}$ we denote the graph
obtained from $G$ by replacing each vertex $v$ by $k$ vertices, say,
$(v,1),\ldots,(v,k)$ and replacing each edge $vw$ of $G$ by a
complete bipartite graph between the $(v,i)$ and $(w,j)$, that is,
adding all edges $(v,i)(w,j)$.

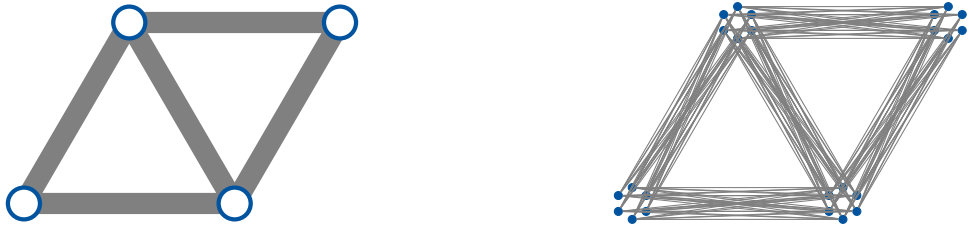
\begin{figure}
  \centering
  \begin{tikzpicture}
    
    \begin{scope}[
      every edge/.style={draw,black!50,line width=8pt},
      ]
      \node (v1) at (90:1.6cm) {}; 
      \node (v2) at (210:1.6cm) {}; 
      \node (v3) at (330:1.6cm) {}; 
      \node (v4) at (30:3.2cm) {};
      
      \draw (v1) edge (v2) edge (v3) edge (v4) (v3) edge (v2) edge
      (v4);

      \draw[blau,ultra thick,fill=white] (90:1.6cm) circle[radius=6pt]
      (210:1.6cm) circle[radius=6pt]
      (330:1.6cm) circle[radius=6pt]
      (30:3.2cm) circle[radius=6pt]
      ;
    \end{scope}

    \begin{scope}[xshift=8cm,
      vertex/.style={blau,circle,draw, minimum size=3pt,inner
        sep=0pt,fill=blau},
      every edge/.style = {draw,black!50},
     ]
     \foreach \i in {0,1,...,5}
     \path (90:1.6cm) +(60*\i+30:6pt) node[vertex] (v1\i) {};
     \foreach \i in {0,1,...,5}
     \path (210:1.6cm) +(60*\i+30:6pt) node[vertex] (v2\i) {};
     \foreach \i in {0,1,...,5}
     \path (330:1.6cm) +(60*\i+30:6pt) node[vertex] (v3\i) {};
     \foreach \i in {0,1,...,5}
     \path (30:3.2cm) +(60*\i+30:6pt) node[vertex] (v4\i) {};

     \foreach \v/\w in {1/2,1/3,1/4,3/2,3/4}
     \foreach \i in {0,1,...,5}
     \foreach \j in {0,1,...,5}
     \draw (v\v\i) edge (v\w\j);
     
    \end{scope}
    
  \end{tikzpicture}
  \caption{A graph $G$ and its blow-up $G^{\odot 6}$}
  \label{fig:blowup}
\end{figure}

The intuitive idea is that if we look from a sufficient distance then we
cannot distinguish the individual copies $(v,i)$ of a vertex $v$, and
the graphs $G$ and $G^{\odot k}$ will look similar (see Figure~\ref{fig:blowup}). We will define their
distance to be $0$. Moreover, to compare two graphs of different
sizes, we blow them up to a common size. Let
$\widehat\delta$ be a suitably normalised graph metric defined on
graphs of the same order, for example, $\widehat\delta_{(1)}, \widehat\delta_{1}, \widehat\delta_{\square}$, and let $G$, $H$
be graphs of order $m,n$, respectively. Then we let
\begin{equation}
  \label{eq:34}
    \delta^\odot(G,H)\coloneqq\lim_{{\ell}\to\infty}\widehat\delta(G^{\odot {\ell}n},H^{\odot {\ell}m}).
\end{equation}
To make sure that $\delta^\odot$ is well-defined, we need to prove that the
limit exists.

\begin{theorem}\label{thm:limit}
  $\delta_{(1)}^\odot, \delta^\odot_1,\delta^\odot_\square$ are well defined, that is,
  the limit in \eqref{eq:34} exists for $\widehat\delta\in\{\widehat\delta_{(1)}, \widehat\delta_1, \widehat\delta_\square\}$.
\end{theorem}

The following lemma is the crucial step for the proof.

  \begin{lemma}\label{lem:limit}
  Let $G,H$ be graphs of the same order and ${\ell}\ge 1$. Then
  \begin{align}
    \label{eq:12}
    \delta_{(1)}(G,H)&\ge \frac{\delta_{(1)}(G^{\odot {\ell}},H^{\odot {\ell}})}{{\ell}^2}\\
    \label{eq:13}
    \delta_{1}(G,H)&\ge \frac{\delta_{1}(G^{\odot {\ell}},H^{\odot {\ell}})}{{\ell}}\\
    \label{eq:14}
    \delta_{\square}(G,H)&\ge
                           \frac{\delta_{\square}(G^{\odot {\ell}},H^{\odot {\ell}})}{{\ell}^2}.
  \end{align}
\end{lemma}

\begin{proof}
  For a matrix $A\in\Real^{I\times I}$, let
  $A^{\odot {\ell}}\in\Real^{(I\times[{\ell}])\times(I\times [{\ell}])}$ be the matrix
  defined
  \[
    A^{\odot {\ell}}\big((i,i'),(j,j')\big)\coloneqq A(i,j).
  \]
  That is, 
  $A^{\odot {\ell}}$ is the tensor product of $A$ with the ${\ell}\times
  {\ell}$-matrix $J_{\ell}$ with all entries $1$. Note that if $A$ is the adjacency matrix of a graph $G$ then
  $A^{\odot {\ell}}$ is the adjacency matrix of $G^{\odot {\ell}}$.

  Without loss of generality we may assume that our graphs $G,H$ have
  the same vertex set $V$. For a permutation $\pi\in\Sym(V)$, we let
  $\pi^{\odot {\ell}}\in\Sym(V\times[{\ell}])$ be defined by
  $\pi^{\odot {\ell}}\big((v,j)\big)\coloneqq \big(\pi(v),j\big)$.

  For every matrix $B\in\Real^{I\times I}$ we obviously have
  $\|B^{\odot {\ell}}\|_{(1)}={\ell}^2\|B\|_{(1)}$. Thus
  \begin{align*}
    \delta_{(1)}(G,H)&=\min_{\pi\in\Sym(V)}\|A_G^\pi-A_H\|_{(1)}\\
    &=\frac{1}{{\ell}^2}\min_{\pi\in\Sym(V)}\|(A_G^{\odot
      {\ell}})^{\pi^{\odot {\ell}}}-A_H^{\odot {\ell}}\|_{(1)}\\
    &\ge \frac{1}{{\ell}^2}\min_{\pi'\in\Sym(V\times[{\ell}])}\|(A_G^{\odot
      {\ell}})^{\pi'}-A_H^{\odot
      {\ell}}\|_{(1)}=\frac{1}{{\ell}^2}\delta_{(1)}(G^{\odot {\ell}},H^{\odot {\ell}}).
  \end{align*}
  This proves \eqref{eq:12}.

  To prove \eqref{eq:13}, we note that for every matrix
  $B\in\Real^{I\times I}$ we have
  \begin{align*}
    \|B^{\odot\ell}\|_{1}&=\max_{(i',j')\in
      I\times[\ell]}\sum_{(i,j)\in
                              I\times[\ell]}B^{\odot\ell}\big((i,j),(i',j')\big)\\
    &=\max_{(i',j')\in
      I\times[\ell]}\sum_{(i,j)\in
      I\times[\ell]}B(i,i')\\
    &=\ell\max_{i'\in I}\sum_{i\in
      I}B(i,i')=\ell\|B\|_1,
  \end{align*}
  where the first and last inequalities hold by Fact~\ref{fact:operator}. Then
  \begin{align*}
    \delta_1(G,H)&=\min_{\pi\in\Sym(V)}\|A_G^\pi-A_H\|_{1}\\
    &=\frac{1}{\ell}\min_{\pi\in\Sym(V)}\|(A_G^{\odot
      {\ell}})^{\pi^{\odot {\ell}}}-A_H^{\odot {\ell}}\|_{1}\\
    &\ge \frac{1}{\ell}\min_{\pi'\in\Sym(V\times[{\ell}])}\|(A_G^{\odot
      {\ell}})^{\pi'}-A_H^{\odot
      {\ell}}\|_{1}=\frac{1}{\ell}\delta_1(G^{\odot {\ell}},H^{\odot {\ell}}).
  \end{align*}
  This proves \eqref{eq:13}.

  To establish \eqref{eq:14}, we shall prove that $\|B^{\odot
    {\ell}}\|_{\square}={\ell}^2\|B\|_{\square}$ for every matrix
  $B\in\Real^{I\times I}$.
  One inequality is trivial:
  \begin{align*}
    \|B^{\odot
    {\ell}}\|_{\square}&=\max_{Y,Y'\subseteq I\times [{\ell}]}\left|\sum_{(i,j)\in
      Y,(i',j')\in Y'}B^{\odot {\ell}}\big((i,j),(i',j')\big)\right|\\
  &\ge \max_{X,X'\subseteq I}\left|\sum_{(i,j)\in
      X\times[{\ell}],(i',j')\in X'\times[{\ell}]}B^{\odot
    {\ell}}\big((i,j),(i',j')\big)\right|\\
    &={\ell}^2\max_{X,X'\subseteq
      I}\left|\sum_{i\in X,i'\in X'}B(i,i')\right|={\ell}^2\|B\|_\square.
  \end{align*}
  To prove the converse inequality, let $Y,Y'\subseteq I\times [{\ell}]$
  such that
  \[
    \|B^{\odot
      {\ell}}\|_{\square}=\left|\sum_{(i,j)\in
        Y,(i',j')\in Y'}B^{\odot {\ell}}\big((i,j),(i',j')\big)\right|.
  \]
  Suppose that there is no $X\subseteq I$ such that
  $Y=X\times[{\ell}]$. Then there are an $i\in I$ and $\emptyset\subset
  J\subset[{\ell}]$ such that $Y\cap \big(\{i\}\times[{\ell}]\big)=\{i\}\times J$. We
  write
  \begin{multline*}
    \sum_{(i,j)\in
      Y,(i',j')\in Y'}B^{\odot {\ell}}\big((i,j),(i',j')\big)
    =\\\underbrace{\sum_{j\in J,(i',j')\in Y'}B^{\odot
        {\ell}}\big((i,j),(i',j')\big)}_{\coloneqq a}+
    \underbrace{\sum_{\substack{(i'',j'')\in Y\text{ with }i''\neq i,\\
          (i',j')\in Y'}}B^{\odot
        {\ell}}\big((i,j),(i',j')\big)}_{\coloneqq b}.
  \end{multline*}
  Then $\|B^{\odot
      {\ell}}\|_{\square}=a+b$.
  Since
  for all $j\in[{\ell}]$ and $(i',j')\in Y'$ we have $B^{\odot
    {\ell}}\big((i,j),(i',j')\big)=B(i,i')$, we have
  \[
    a=|J|\cdot\underbrace{\sum_{(i',j')\in Y'}B(i,i')}_{\coloneqq a'}.
  \]
  If $a$ and $b$ have the same sign, then with $Z\coloneqq
  Y\cup\big(\{i\}\times[{\ell}]\big)$ we have
  \begin{multline*}
    \left|\sum_{(i,j)\in
        Z,(i',j')\in Y'}B^{\odot
        {\ell}}\big((i,j),(i',j')\big)\right|\\
    =\left|\sum_{(i,j)\in
        Y,(i',j')\in Y'}B^{\odot
        {\ell}}\big((i,j),(i',j')\big)\right|+({\ell}-|J|)|a'|\\\ge \left|\sum_{(i,j)\in
        Y,(i',j')\in Y'}B^{\odot
        {\ell}}\big((i,j),(i',j')\big)\right|.
  \end{multline*}
  If $a$ and $b$ have different signs, then with $Z\coloneqq
  Y\setminus\big(\{i\}\times[{\ell}]\big)$ we have
  \begin{multline*}
    \left|\sum_{(i,j)\in
        Z,(i',j')\in Y'}B^{\odot
        {\ell}}\big((i,j),(i',j')\big)\right|\\
    =\left|\sum_{(i,j)\in
        Y,(i',j')\in Y'}B^{\odot
        {\ell}}\big((i,j),(i',j')\big)\right|+|J|\cdot|a'|\\\ge \left|\sum_{(i,j)\in
        Y,(i',j')\in Y'}B^{\odot
        {\ell}}\big((i,j),(i',j')\big)\right|.
  \end{multline*}
  In both cases, we can replace $Y$ by $Z$, and possibly repeating
  the argument we may assume that $Y=X\times[{\ell}]$ for some $X\subseteq
  I$. Similarly, we may assume that $Y'=X'\times[{\ell}]$ for some
  $X'\subseteq I$.
  Then
    \begin{align*}
    \|B^{\odot
    {\ell}}\|_{\square}&=\left|\sum_{(i,j)\in
     Y,(i',j')\in Y'}B^{\odot
                         {\ell}}\big((i,j),(i',j')\big)\right|\\
      &=\left|\sum_{(i,j)\in
      X\times[{\ell}],(i',j')\in X'\times[{\ell}]}B^{\odot {\ell}}\big((i,j),(i',j')\big)\right|\\
  &={\ell}^2\left|\sum_{i\in
      X,i'\in X'}B(i,i')\right|\\
    &\le {\ell}^2\max_{X,X'\subseteq
      I}\left|\sum_{i\in X,i'\in X'}B(i,i')\right|={\ell}^2\|B\|_\square.
  \end{align*}
This proves $\|B^{\odot
    {\ell}}\|_{\square}={\ell}^2\|B\|_{\square}$, and \eqref{eq:14} follows. 
\end{proof}

\begin{proof}[Proof of Theorem~\ref{thm:limit}]
  We give the proof for $\delta_{(1)}^\odot$. The proofs for
  $\delta_1^\odot$ and $\delta_\square^\odot$ are similar.

  Let 
  $G_1\coloneqq G^{\odot n}$, $H_1\coloneqq H^{\odot m}$. Note that
  $G_1,H_1$ are graphs of the same order $n_1\coloneqq mn$. Let
  $V_1\coloneqq V_G\times[n]$ and $W_1\coloneqq V_H\times[m]$ be the
  vertex sets of $G_1,H_1$, respectively. Note that for all $\ell\ge 1$ the graph $G^{\odot\ell n}$ with vertex
  set $V_G\times[\ell n]$ and the graph
  $G_\ell\coloneqq G_1^{\odot\ell}$ with vertex set
  $V_1\times[\ell]=V_G\times[n]\times[\ell]$ are isomorphic.
  For $\ell\ge 1$, let $A_\ell$ be the adjacency matrix of
  $G_\ell$. Let $H_\ell\coloneqq H_1^{\odot\ell}$ and observe that
  $H_\ell\cong H^{\odot\ell m}$. Moreover, let $B_\ell$ be
  its adjacency matrix of $H_\ell$, and let $d_\ell\coloneqq \delta_{(1)}(G_\ell,H_\ell)$. We
  need to prove that the limit
  $
    \lim_{\ell\to\infty}\frac{d_\ell}{\ell^2}
  $
  exists.

  It follows from Lemma~\ref{lem:limit} that for all $j,k\in\PNat$
  it holds that
  \begin{equation}
    \label{eq:36}
    \frac{d_k}{k^2}\ge\frac{d_{jk}}{(jk)^2}. 
  \end{equation}
  Furthermore, for all $\ell\in\Nat$ it holds that
  \begin{equation}
    \label{eq:37}
   d_{\ell+1}-d_\ell\le d_1.
  \end{equation}
  To see this, let $\pi\in\Bij(V_1\times[\ell],W_1\times[\ell])$ such that
  $d_\ell=\big\|A_\ell^\pi-B_\ell\big\|_{(1)}$, and $\pi'\in\Bij(V_1,W_1)$ such that
  $d_1=\big\|A_1^{\pi'}-B_1\big\|_{(1)}$. We define
  $\pi''\in\Bij(V_1\times[\ell+1],W_1\times[\ell+1])$ by
  \[
    \pi''(v,i)\coloneqq
    \begin{cases}
      \pi(v,i)&\text{if }i\in[\ell],\\
      \big(\pi'(v),\ell+1\big)&\text{if }i=\ell+1.
    \end{cases}
  \]
  Then
  \[
    d_{\ell+1}\le\big\|A_{\ell+1}^{\pi''}-B_{\ell+1}\big\|_{(1)}\le \big\|A_\ell^\pi-B_\ell\big\|_{(1)}+\big\|A_1^{\pi'}-B_1\big\|_{(1)}=d_\ell+d_1,
  \]
  which implies \eqref{eq:37}.

  It follows from \eqref{eq:36} (with $k=1$) that
  $\big(d_\ell/\ell^2\big)_{\ell\ge 1}$ is a bounded sequence of
  non-negative reals. Hence
  \[
    d\coloneqq\liminf_{\ell\to\infty}\frac{d_\ell}{\ell^2}
  \]
  exists. We shall prove that
  \begin{equation}
    \label{eq:39}
    \lim_{\ell\to\infty}\frac{d_\ell}{\ell^2}=d.
  \end{equation}
  Let $\epsilon>0$. By the definition of $d$ there is a $k\in\Nat$ such that
  \begin{equation}
    \label{eq:41}
    \frac{d_k}{k^2}\le d+\frac{\epsilon}{2}
  \end{equation}
  and
  \begin{equation}
    \label{eq:42}
     \frac{d_\ell}{\ell^2}\ge d-\epsilon
  \end{equation}
  for all $\ell\ge k$. We fix such a $k$ and choose $j_0\in\PNat$ such
  that for $\ell_0\coloneqq j_0k$ we have
  \[
    j_0k\ge\sqrt{\frac{2d_1k}{\epsilon}}.
  \]
  Then 
  $\frac{d_1}{(j_0k)^2}\le\frac{\epsilon}{2k}$. Hence, by
  \eqref{eq:37}, for all $i\ge j_0k$
  \begin{equation}
    \label{eq:40}
    \frac{d_{i+1}}{(i+1)^2}-\frac{d_{i}}{i^2}\le\frac{d_{i+1}-d_i}{i^2}\le
    \frac{d_1}{(j_0k)^2}\le\frac{\epsilon}{2k}.
  \end{equation}
  Let $\ell_0\coloneqq j_0k$. We shall prove that for all $\ell\ge
  \ell_0$ we have
  \[
    \left|\frac{d_\ell}{\ell^2}-d\right|\le\epsilon.
  \]
  Note that $\frac{d_\ell}{\ell^2}\ge d-\epsilon$ by
  \eqref{eq:42}. Thus we only need to prove
  \begin{equation}
    \label{eq:43}
    \frac{d_\ell}{\ell^2}\le d+\epsilon.
  \end{equation}
  Let $j\coloneqq \floor{\frac{\ell}{k}}$. Then $j\ge
  j_0$. It follows from \eqref{eq:40} that 
  \begin{equation*}
    \label{eq:44}
        \frac{d_{\ell}}{\ell^2}-\frac{d_{jk}}{(jk)^2}\le(\ell-jk) \frac{\epsilon}{2k}\le\frac{\epsilon}{2}.
  \end{equation*}
  By \eqref{eq:36}, we have
  \[
    \frac{d_{jk}}{(jk)^2}\le \frac{d_k}{k^2}.
  \]
  Overall,
  \[
    \frac{d_{\ell}}{\ell^2}\le
    \frac{d_{jk}}{(jk)^2}+\frac{\epsilon}{2}\le
    \frac{d_k}{k^2}+\frac{\epsilon}{2}\le d+\epsilon.
  \]
  This proves \eqref{eq:43}.
\end{proof}

Let $\delta\in\{\delta_{(1)}, \delta_1, \delta_\square\}$. By Lemma~\ref{lem:limit},
for graphs $G,H$ of the same order $n$ it holds that
\[
  \delta^\odot(G,H)\le\widehat\delta(G^{\odot {\ell}},H^{\odot {\ell}})\le\widehat\delta(G,H).
\]
The following example illustrates that the inequalities may be strict
for $\delta_{(1)}$. Similar examples are known for $\delta_1$ (due to Timo
Gervens, private communication) and $\delta_\square$
\cite[Section 8.1, Exercise~8.8]{Lovasz12}.

\begin{example}[\cite{GoldreichKNR08}]\label{exa:limit}
  This example is attributed to Arie Matsliah in
  \cite{GoldreichKNR08}.
  
  Consider the two 4-vertex graphs graphs $G=([4],\{12,13,23\}\big)$
  and $H=\big([4],\{12,34\})$ (see Figure~\ref{fig:limit}(a)). It is not hard to see that
  $\ged(G,H)=3$, because no matter how we choose a bijection $\pi$
  between the vertex sets, two of the edges of the triangle in $G$
  will be mapped to non-edges in $H$, and one of the two edges of $H$
  will be the image of a non-edge in $G$. Hence $\delta_{(1)}(G,H)=6$.
  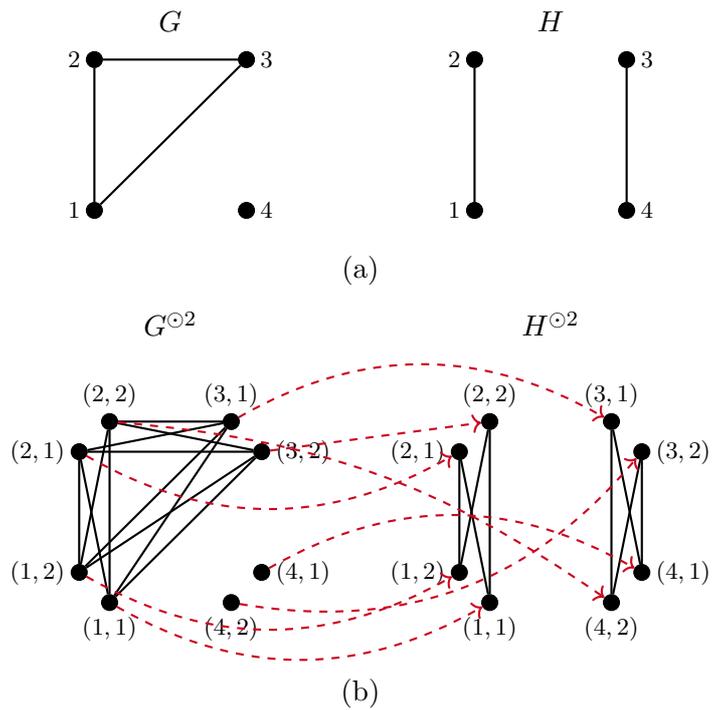
\begin{figure}\centering
  \begin{tikzpicture}[
    vertex/.style={circle,draw,fill,inner sep=0pt,minimum height=6pt},
    ]
    \footnotesize
    \begin{scope}
      \node[vertex] (v1) at (0,0) {};
      \node[vertex] (v2) at (0,2) {};
      \node[vertex] (v3) at (2,2) {};
      \node[vertex] (v4) at (2,0) {};

      \path (v1) node[left=2pt] {$1$} (v2) node[left=2pt] {$2$} (v3)
      node[right=2pt] {$3$} (v4)  node[right=2pt] {$4$};

      \draw[thick] (v1) edge (v2) edge (v3) (v2) edge (v3);

      \node[vertex] (w1) at (5,0) {};
      \node[vertex] (w2) at (5,2) {};
      \node[vertex] (w3) at (7,2) {};
      \node[vertex] (w4) at (7,0) {};

      \path (w1) node[left=2pt] {$1$} (w2) node[left=2pt] {$2$} (w3)
      node[right=2pt] {$3$} (w4)  node[right=2pt] {$4$};

      \draw[thick] (w1) edge (w2) (w3) edge (w4);

      \path (1,2.5) node {\normalsize$G$} (6,2.5) node
      {\normalsize$H$} (3.5,-0.8) node {\normalsize (a)};
     \end{scope}

     \begin{scope}[yshift=-5cm]
       \node[vertex] (v11) at (0.2,-0.2) {}; 
       \node[vertex] (v12) at (-0.2,0.2) {}; 
       \node[vertex] (v21) at (-0.2,1.8) {}; 
       \node[vertex] (v22) at (0.2,2.2) {}; 
       \node[vertex] (v31) at (1.8,2.2) {}; 
       \node[vertex] (v32) at (2.2,1.8) {}; 
       \node[vertex] (v41) at (2.2,0.2) {}; 
       \node[vertex] (v42) at (1.8,-0.2) {}; 

       \path (v11) node[below=2pt] {$(1,1)$} (v12) node[left=2pt]
       {$(1,2)$}
       (v21) node[left=2pt] {$(2,1)$} (v22) node[above=2pt] {$(2,2)$}
       (v31) node[above=2pt] {$(3,1)$} (v32) node[right=2pt] {$(3,2)$}
       (v41) node[right=2pt] {$(4,1)$} (v42) node[below=2pt] {$(4,2)$}
       ;

      \draw[thick] (v11) edge (v21) edge (v22) edge (v31) edge (v32)
      (v12) edge (v21) edge (v22) edge (v31) edge (v32)  (v21) edge
      (v31) edge (v32) (v22) edge (v31) edge (v32);

      \node[vertex] (w11) at (5.2,-0.2) {}; 
      \node[vertex] (w12) at (4.8,0.2) {}; 
      \node[vertex] (w21) at (4.8,1.8) {}; 
      \node[vertex] (w22) at (5.2,2.2) {}; 
       \node[vertex] (w31) at (6.8,2.2) {}; 
       \node[vertex] (w32) at (7.2,1.8) {}; 
       \node[vertex] (w41) at (7.2,0.2) {}; 
       \node[vertex] (w42) at (6.8,-0.2) {}; 

       \path (w11) node[below=2pt] {$(1,1)$} (w12) node[left=2pt]
       {$(1,2)$}
       (w21) node[left=2pt] {$(2,1)$} (w22) node[above=2pt] {$(2,2)$}
       (w31) node[above=2pt] {$(3,1)$} (w32) node[right=2pt] {$(3,2)$}
       (w41) node[right=2pt] {$(4,1)$} (w42) node[below=2pt] {$(4,2)$}
       ;

       \draw[thick] (w11) edge (w21) edge (w22) (w12) edge (w21) edge
       (w22) (w31) edge (w41) edge (w42) (w32) edge (w41) edge (w42);

       \draw[thick,dashed,rot,->] (v11) edge[bend right] (w11)
       (v12) edge[bend right] (w12)
       (v21) edge[bend right] (w21)
       (v22) edge[bend left=15] (w42)
       (v31) edge[bend left] (w31)
       (v32) edge (w22)
       (v41) edge[bend left] (w41)
      (v42) edge[bend right] (w32)
       ;

      \path (1,3.5) node {\normalsize$G^{\odot2}$} (6,3.5) node {\normalsize$H^{\odot2}$}  (3.5,-1.4) node {\normalsize (b)};
     \end{scope}
 
  \end{tikzpicture}
  \caption{The graphs of Example~\ref{exa:limit}: (a) shows the graphs
  $G,H$, (b) shows their blow-ups $G^{\odot2},H^{\odot2}$ and the
  bijection $\pi$ between them}\label{fig:limit}
\end{figure}
  Now consider the blown-up graphs $G^{\odot2}$ and $H^{\odot2}$. Let
  $\pi:[4]\times[2]\to [4]\times[2]$ be the following bijection
  between their vertex sets:
  \[
    \begin{array}{c|c|c|c|c|c|c|c|c}
      (v,i)   &(1,1)&(1,2)&(2,1)&(2,2)&(3,1)&(3,2)&(4,1)&(4,2)\\
      \hline
      \pi(v,i)&(1,1)&(1,2)&(2,1)&(4,2)&(3,1)&(2,2)&(4,1)&(3,2)
    \end{array}
  \]
  (see Figure~\ref{fig:limit}(b)).
  A straightforward calculation shows that
  $\|A_{G^{\odot2}}^\pi-A_{H^{\odot2}}\|_{(1)}=20$.
  Indeed, we have
  \[
   A_{G^{\odot2}}^\pi=
      \begin{pmatrix}
        0&0&1&1&1&0&0&1\\
        0&0&1&1&1&0&0&1\\
        1&1&0&1&1&0&0&0\\
        1&1&1&0&0&0&0&1\\
        1&1&1&0&0&0&0&1\\
        0&0&0&0&0&0&0&0\\
        0&0&0&0&0&0&0&0\\
        1&1&0&1&1&0&0&0\\
      \end{pmatrix}
   \quad\text{and}\quad
   A_{H^{\odot2}}=
      \begin{pmatrix}
        0&0&1&1&0&0&0&0\\
        0&0&1&1&0&0&0&0\\
        1&1&0&0&0&0&0&0\\
        1&1&0&0&0&0&0&0\\
        0&0&0&0&0&0&1&1\\
        0&0&0&0&0&0&1&1\\
        0&0&0&0&1&1&0&0\\
        0&0&0&0&1&1&0&0
      \end{pmatrix}.
     \]
     Comparing these two matrices, we find the following numbers of
     mismatches per column:
     \[
       \begin{array}{c|c|c|c|c|c|c|c|c}
         \text{col}&1&2&3&4&5&6&7&8\\
         \hline
         \text{\# mismatches}&2&2&2&2&4&2&2&4
       \end{array}.
     \]
     Hence overall, we have 20 mismatches, which means that
     $\|A_{G^{\odot2}}^\pi-A_{H^{\odot2}}\|_{(1)}=20$.
     
  Thus
    $\delta_{(1)}(G^{\odot2},H^{\odot2})\le20$ and therefore, by
    Lemma~\ref{lem:limit},
    $\frac{\delta_{(1)}(G^{\odot2\ell},H^{\odot2\ell})}{\ell^2}\le 20$ for
    every $\ell\ge 1$. Hence 
    \[
      \delta_{(1)}^{\odot}(G,H)\le\widehat\delta_{(1)}(G^{\odot2},H^{\odot2})\le\frac{20}{64}<\frac{6}{16}=\widehat\delta_{(1)}(G,H).\uende
    \]
\end{example}

So it may happen that $\delta^\odot(G,H)<\widehat\delta(G,H)$. However, at least
for $\delta_{(1)}$ and $\delta_{\square}$ we also have a lower bound
for $\delta^\odot(G,H)$ in terms of $\widehat\delta(G,H)$, as the
following theorem shows.

\begin{theorem}[Borgs et al.~\cite{BorgsCLSV08},
  Pikhurko~\cite{Pikhurko10}]
  \label{thm:dot}
  For all graphs  $G,H$ of the same order we have
\begin{align}
  \label{eq:180}
  \widehat\delta_{(1)}(G,H)&\le
    3\delta_{(1)}^\odot(G,H),\\
  \label{eq:18}
  \widehat\delta_\square(G,H)&\le
                                    32\delta_\square^\odot(G,H)^{1/67}.
\end{align}
\end{theorem}

We will prove assertion \eqref{eq:180} for edit distance in
Section~\ref{sec:got}. For a proof of \eqref{eq:18}, we refer the
reader to \cite{BorgsCLSV08}.

\subsection{Relaxations}
\label{sec:relax}
To compute any of the graph distances discussed so far, at some point
we need to minimise over a set of permutations (or bijections), which
is a notoriously hard computational problem (see
Section~\ref{sec:complexity}). In this section, we address this problem by replacing permutations with simpler
forms of alignment between the graphs.

Let $\|\cdot\|$ be a permutation-invariant matrix norm, and let
$G=(V,E_G)$ and $H=(W,E_H)$ be graphs of the same order $n$. Then
\[
  \delta_{\|\cdot\|}(G,H)=\min_{P\in\Perm(V,W)}\|P^\top A_GP-A_H\|=\min_P\|P^{-1}A_GP-A_H\|=\min_P\|A_GP-PA_H\|,
\]
where $\Perm(V,W)$ denotes the set of all $V\times W$ permutation matrices. In the
second equality we use that permutation matrices $P$ are orthogonal,
that is $P^\top=P^{-1}$, and in the third equality we use that
$\|\cdot\|$ is permutation invariant. The last expression
$\min_P\|A_GP-PA_H\|$ has the advantage of having a linear term inside
the norm. Now we look at the relaxation
\begin{equation}
  \label{eq:6}
    \min_{Q\in\DS(V,W)}\|A_GQ-QA_H\|
\end{equation}
where $\DS(V,W)$ denotes the set of all doubly stochastic
$V\times W$ matrices, that is, all matrices $Q\in\NNReal^{V\times W}$ with row and
column sums $1$. Note that the integer doubly stochastic matrices are
precisely the permutation matrices. By Birkhoff's Theorem, the
doubly stochastic matrices are precisely the convex combinations of
permutation matrices. In particular, $\DS(V,W)$ is a convex subset of
$\Real^{V\times W}$, and \eqref{eq:6} describes a convex optimisation
problem, which can be solved efficiently by gradient-descent methods
(provided that we can efficiently compute $\|A\|$ for a matrix $A$ and
that it is differentiable).
To see that the function $Q\mapsto \|A_GQ-QA_H\|$, for $Q\in\DS(V,W)$,
is indeed convex, let $\alpha_1,\alpha_2\in\NNReal$ with
$\alpha_1+\alpha_2=1$ and $Q_1,Q_2\in\DS(V,W)$. Then 
\begin{align*}
  &\big\|A_G(\alpha_1 Q_1+\alpha_2 Q_2)-(\alpha_1 Q_1+\alpha_2 Q_2)A_H\big\|\\
  =\;&\big\|\alpha_1 (A_GQ_1-Q_1A_H)+\alpha_2(A_GQ_2-Q_2A_H)\big\|\\
  \le\;&\alpha_1\|A_GQ_1-Q_1A_H\|+\alpha_2\|A_GQ_2-Q_2A_H\|
\end{align*}
by the properties of a norm.

We can use \eqref{eq:6} to define a graph metric
on graphs of the same order. We generalise the definition to graphs $G,H$ of
different orders $m\coloneqq|V|,n\coloneqq|W|$ by replacing doubly stochastic
$V\times W$ matrices by nonnegative $V\times W$ matrices with row sums
$1/m$ and column sums $1/n$. Let $\DS^*(V,W)$ denote the (convex) set
of all such matrices. Then we let
\begin{equation}
  \label{eq:7}
    \delta_{\|\cdot\|}^*(G,H)\coloneqq \min_{Q\in\DS^*(V,W)}\left\|\frac{1}{m}A_GQ-\frac{1}{n}QA_H\right\|. 
\end{equation}
Here the minimum exists by a simple compactness argument (the same
holds for similar definitions that follow).
We call
$\delta^*_{\|\cdot\|}$ the \emph{fractional relaxation} of
$\delta_{\|\cdot\|}$. We denote the fractional relaxations of
$\delta_{(p)},\delta_p,\delta_{\square}$ by
$\delta^{*}_{(p)},\delta^*_{p},\delta^*_{\square}$, respectively. We
can also define a fractional relaxation $\delta^*_{\textup{dist}}$ of 
$\delta_{\textup{dist}}$. The reader may wonder about the
normalisation factors $\frac{1}{m}$ and $\frac{1}{n}$ in
\eqref{eq:7}. They are chosen this way to factor out the different
graph sizes. We may have to adapt them depending on the matrix norm we
use. Observe that for graphs of the same order $m=n$ we have
$\delta_{\|\cdot\|}^*(G,H)=\frac{1}{n^2}\min_{Q\in\DS(V,W)}\|A_GQ-QA_H\|$,
  which brings back the normalisation factor $1/n^2$ that we use for
  many norms, most notably the entrywise $\ell_1$-norm and the cut norm.

Besides the convex relaxation leading to the metrics $\delta^*$ there
is a second form of relaxation that is particularly relevant for the
distortion distance $\delta_{\textup{dist}}$. A \emph{correspondence}
between sets $V$ and $W$ is a relation $R\subseteq V\times W$ such
that for all $v\in V$ there is a $w\in W$ such that $(v,w)\in R$ and
for all $w\in W$ there is a $v\in V$ such that $(v,w)\in R$. We denote
the set of all correspondences between $V$ and $W$ by
$\Corr(V,W)$. Then we define the \emph{Gromov-Hausdorff distance}
between graphs $G=(V,E_G), H=(W,E_H)$ by
\[
  \delta_{\textup{GH}}(G,H)\coloneqq\min_{R\in\Corr(V,W)}\max_{(v,w),(v',w')\in
    R}\big| D_G(v,v')-D_H(w,w')\big|.
\]
Observe that for sets $V,W$ of the same order, the graph of a
bijection in $\Bij(V,W)$ is in $\Corr(V,W)$. For graphs $G,H$ of the
same order, the distortion distance is defined in the same way as the Gromov-Hausdorff
distance, except that we only minimise over correspondences that are graphs of
bijections. Originally, the Gromov-Hausdorff distance is defined on arbitrary
metric spaces in a different way. For this, let $U$ be a set
and let $d$ be a metric on $U$. For simplicity, let us assume here that $U$ is
finite.\footnote{More generally, we could take $(U,d)$ to be a compact metric space.} For a subset $V\subseteq U$ and an element $w\in U$ we let
$d(V,w)\coloneqq d(w,V)\coloneqq \min_{v\in V}d(v,w)$. The \emph{Hausdorff distance} between subsets $V,W\subseteq U$
is defined to be $\max\big\{\max_{v\in V}d(v,W),\; \max_{w\in W}d(V,w)\big\}$.
Then it can be shown that for our graphs $G,H$, $\frac{1}{2}\delta_{\textup{GH}}(G,H)$ is the minimum
Hausdorff distance between their vertex sets $V$ and $W$  with respect to a
metric $d$ on the disjoint union of $V$ and $W$ such that the
restriction of $d$ to $V$ is the shortest path metric of $G$ and the
restriction of $d$ to $W$ is the shortest path metric of $H$ (see
\cite{Memoli11}).

\subsection{Graph Optimal Transport}
\label{sec:got}
The best known example of an \emph{optimal transport distance} is the
\emph{Wasserstein distance}, a.k.a. \emph{Earth movers distance}
measuring the distance between two probability spaces (or more generally
measured spaces) jointly embedded into a metric space. Let us briefly
explain this in the simple case of finite spaces. Suppose we have
probability measures $p_V,p_W$ on finite sets $V,W$. A \emph{coupling}
between $p_V$ and $p_W$ is a probability measure $q$ on $V\times W$ whose
marginal distributions on $V,W$ coincide with $p_V,p_W$,
respectively. That is, for all $v\in V$ it holds that
$\sum_{w\in W}q(v,w)=p_V(v)$ and for all $w\in W$ it holds that
$\sum_{v\in V}q(v,w)=p_W(w)$. In the following, we view such a
coupling as an $V\times W$ matrix with row sums given by $p_V$ and
column sums given by $p_W$. By $\Cp(p_V,p_W)$ we denote the set of
all such matrices. Note that if $p_V$ and $p_W$ are the
uniform distributions on $V,W$, respectively, then
$\Cp(p_V,p_W)=\DS^*(V,W)$. Now suppose that we have a metric space $(U,d)$
with $V,W\subseteq U$. Then the \emph{Wasserstein distance} between
$p_V$ and $p_W$ (with respect to $d$) is
$\min_{Q\in\Cp(p_V,p_W)}\sum_{v\in V,w\in
  W}d(v,w)Q(v,w)$. Intuitively, this measures the amount of
probability mass per distance unit that we have to move to transform $p_V$
to $p_W$.

Let us combine the Wasserstein distance with our previous graph
metrics. For this, suppose we have graphs $G=(V,E_G)$ and
$H=(W,E_H)$ of orders $m\coloneqq|V|,n\coloneqq|W|$ and probability
measures $p_G$ on $V$ and $p_H$ on $W$. These probability measures may
be given with the graphs by explicit weight functions on the vertices,
or they may be derived from the graphs by assigning some kind of
``importance'' to the vertices.  For example, in \cite{MareticGCF19},
Maretic et al.\ propose to take a normal distribution with the pseudo-inverse
of the Laplacian matrix of a graph as the covariance matrix to define
a probability distribution on the vertex set of a graph.  But even the
simplest case of just taking the uniform distribution on the vertices is
interesting. In fact, it is the most interesting case for us
here.

Once we have probability distributions on the vertex sets of our
graphs, the idea is to align the two graphs via a  coupling rather
than a bijection, doubly stochastic
matrix, or correspondence, as in our previous definitions. If we apply this idea to the distortion distance or the
Gromov-Hausdorff distance, we get the following distance known as the
\emph{Gromov-Wasserstein} distance (see~\cite{Memoli11}):
\[
  \delta_{\textup{dist}}^{\textup{OT}}\big((G,p_G),(H,p_H)\big)
  \coloneqq \min_{Q\in\Cp(p_G,p_H)}\sum_{v,v'\in V,w,w'\in
    W}\hspace{-3mm}\big|D_G(v,v')-D_H(w,w')\big|Q(v,w)Q(v',w').
\]
Similarly, we can define optimal transport versions of other graph
metrics, for example:
\begin{align*}
  \delta_{(1)}^{\textup{OT}}\big((G,p_G),(H,p_H)\big)
  &\coloneqq \min_{Q\in\Cp(p_G,p_H)}\sum_{v,v'\in V,w,w'\in
    W}\hspace{-3mm}\big|A_G(v,v')-A_H(w,w')\big|Q(v,w)Q(v',w'),\\
  \delta_{\square}^{\textup{OT}}\big((G,p_G),(H,p_H)\big)
  &\coloneqq \min_{Q\in\Cp(p_G,p_H)}\left|\sum_{v,v'\in V,w,w'\in
    W}\hspace{-3mm}\big(A_G(v,v')-A_H(w,w')\big)Q(v,w)Q(v',w')\right|.
\end{align*}
(It is not so clear how to define an optimal transport version of
metrics such as $\delta_1$ that are based on operator norms.)
If $p_G,p_H$ are the uniform distributions, we omit
them in the notation and just write
$\delta_{\cdots}^{\textup{OT}}(G,H)$. Note that then we minimise over
$\DS^*(V,W)$.

Since for both the fractional and the optimal transport metrics we
minimise over matrices in $\DS^*(V,W)$, the reader may wonder if the
two types of metrics are related.

\begin{proposition}
  For all graphs $G,H$ it holds that
  \begin{align}
  \label{eq:8}
  \delta_{(1)}^*(G,H)&\le\delta_{(1)}^{\textup{OT}}(G,H),\\
  \label{eq:8a}
  \delta_{\square}^*(G,H)&\le\delta_{\square}^{\textup{OT}}(G,H).
\end{align}
\end{proposition}

\begin{proof}
  We only prove
\eqref{eq:8}, the proof of \eqref{eq:8a} is similar.

Let $G=(V,E_G)$, $H=(W,E_H)$ be graphs and
$m\coloneqq|V|,n\coloneqq|W|$ their orders. 
Moreover, let
$Q\in\DS^*(V,W)$ be such that
\[
  \delta_{(1)}^{\textup{OT}}\big(G,H)\big)
  =\sum_{v,v'\in V,w,w'\in
    W}\hspace{-3mm}\big|A_G(v,v')-A_H(w,w')\big|Q(v,w)Q(v',w').
\]
Observe
that for all $v\in V,w'\in W$ we have
\begin{align*}
  \sum_{v'\in V,w\in W}A_G(v,v')Q(v,w)Q(v',w')&=\sum_{v'\in
                                                V}A_G(v,v')Q(v',w')\underbrace{\sum_{w\in W}Q(v,w)}_{=\frac{1}{m}}\\
  &=\frac{1}{m}(A_GQ)(v,w').
\end{align*}
Similarly,
\[
  \sum_{v'\in V,w\in W}A_H(w,w')Q(v,w)Q(v',w')=\frac{1}{n}(QA_H)(v,w').
\]
Then
\begin{align*}
  \delta_{(1)}^{\textup{OT}}\big(G,H)\big)
  &=\sum_{v,v'\in V,w,w'\in
    W}\hspace{-3mm}\big|A_G(v,v')-A_H(w,w')\big|Q(v,w)Q(v',w')\\
  &=\sum_{v\in V,w'\in W}\sum_{v'\in V,w\in
    W}\big|A_G(v,v')Q(v,w)Q(v',w')-A_H(w,w')
    Q(v,w)Q(v',w')\big|\\
  &\ge \sum_{v\in V,w'\in W}\left|\sum_{v'\in V,w\in
    W}\big(A_G(v,v')Q(v,w)Q(v',w')-A_H(w,w')
    Q(v,w)Q(v',w')\big)\right|\\
  &=\sum_{v\in V,w'\in
    W}\left|\frac{1}{m}(A_GQ)(v,w')-\frac{1}{n}(QA_H)(v,w')\right|\\
  &=\left\|\frac{1}{m}A_GQ-\frac{1}{n}QA_h\right\|_{(1)}\\
  &\ge\delta^*_{(1)}(G,H).
\end{align*}
\end{proof}

Conversely, it is not possible to (multiplicatively) bound
$\delta_{(1)}^{\textup{OT}}$ or $\delta_{\square}^{\textup{OT}}$ in
terms of the respective fractional metrics $\delta_{(1)}^*$, $\delta_{\square}^*$, as the
following example shows.

\begin{example}
  Let $G=\big([6],\{12,23,34,45,56,61\}\big)$ be a cycle of length $6$
  and $H=\big([6],\{12,23,31,45,56,64\}\big)$ the union of two
  triangles.

  Then $\delta_{(1)}(G,H), \delta_{\square}(G,H)>0$ and hence $\delta_{(1)}^\odot(G,H), \delta_{\square}^\odot(G,H)>0$ by
  Theorem~\ref{thm:dot}.

  However, for the doubly stochastic matrix $Q\in\Real^{[6]\times[6]}$
  with all entries $Q(v,v')=\frac{1}{6}$ for all $v,v'\in[6]$ it holds
  that $A_GQ=QA_H$. Hence $\delta_{(1)}^*(G,H)=\delta_{\square}^*(G,H)=0$.\uend
\end{example}

Less obviously, there is an exact correspondence between
the optimal transport and the blow-up versions of our metrics.

\begin{theorem}\label{thm:dotot}
  For all graphs
  $G,H$ we have
  \begin{align}
    \label{eq:81}
    \delta_{(1)}^\odot(G,H)&=\delta_{(1)}^{\textup{OT}}(G,H),\\
    \label{eq:81a}
    \delta_{\square}^\odot(G,H)&=\delta_{\square}^{\textup{OT}}(G,H).
  \end{align}
\end{theorem}

We only prove \eqref{eq:81} here. The proof of \eqref{eq:81a} is similar (see \cite{Lovasz12}).
We need the following simple lemma.

\begin{lemma}
  Let $I,J$ be finite sets and $Q\in\DS^*(I,J)$. Then for every
  $\epsilon>0$ there is a matrix $Q'\in\DS^*(I,J)$ with rational
  entries such that $\|Q-Q'\|_{(\infty)}\le\epsilon$.
\end{lemma}

\begin{proof}
   Let $m\coloneqq|I|,n\coloneqq|J|$, and $\epsilon>0$. For all $(i,j)\in I\times J$, let
   $a_{ij},b_{ij}\in\Rat$ such that $Q(i,j)-\epsilon\le a_{ij}\le
   Q(i,j)\le b_{ij}\le Q(i,j)+\epsilon$. Consider the following system
  $L$ of linear inequalities in the variables $x_{ij}$, for $i\in
  I,j\in J$:
  \begin{align*}
    X_{ij}&\ge a_{ij}&\text{for all }i\in I,j\in J,\\ 
    X_{ij}&\le b_{ij}&\text{for all }i\in I,j\in J,\\ 
    X_{ij}&\ge 0&\text{for all }i\in I,j\in J,\\
    \sum_{j=1}^nX_{ij}&=\frac{1}{m}&\text{for all }i\in I,\\
    \sum_{i=1}^mX_{ij}&=\frac{1}{n}&\text{for all }j\in J.
  \end{align*}
  Then setting $X_{ij}$ to $Q(i,j)$ for all $i,j$ yields as solution
  to this system, and every solution corresponds to a matrix
  $Q'\in\DS^*(I,J)$ with $\|Q-Q'\|_{(\infty)}\le\epsilon$. Since a solvable system of
  linear inequalities with rational coefficients always has a rational
  solution, the assertion of the lemma follows.
\end{proof}

\begin{corollary}
  For all graphs $G=(V,E_G)$ and $H\coloneqq(W,E_H)$ and all
  $\epsilon>0$ there is a
  rational matrix $Q\in\DS^*(V,W)$ such that
  \[
    \sum_{v,v'\in V,w,w'\in
      W}\hspace{-3mm}\big|A_G(v,v')-A_H(w,w')\big|Q(v,w)Q(v',w')\le\delta_{(1)}^{\textup{OT}}(G,H)+\epsilon.
  \]
\end{corollary}

\begin{proof}[Proof of \eqref{eq:81} (Theorem~\ref{thm:dotot})]
Let $G=(V,E_G)$ and $H=(W,E_H)$ be graphs of order
  $m\coloneqq|V|,n\coloneqq|W|$. 
  We first prove that
  \[
    \delta_{(1)}^\odot(G,H)\ge\delta_{(1)}^{\textup{OT}}(G,H).
  \]
  Let
  ${\ell}\in\Nat$ and $G'\coloneqq G^{\odot{\ell}n},H'\coloneqq H^{\odot{\ell}m}$, and let
  $N\coloneqq {\ell}mn=|G'|=|H'|$. We shall
  prove that
  \begin{equation}
    \label{eq:15}
    \frac{1}{N^2}\delta_{(1)}(G',H')\ge \delta_{(1)}^{\textup{OT}}(G,H).
  \end{equation}
  Recall that the vertex set of $G'$ is $V\times[{\ell}n]$ and the vertex
  set of $H'$ is $W\times[{\ell}m]$.
  Let $\pi\in\Bij(V\times[{\ell}n],W\times [{\ell}m])$ such that
  $\delta_{(1)}(G',H')=\|A_{G'}^\pi-A_{H'}\|_{(1)}$. We define a matrix
  $Q\in\Real^{V\times W}$ by
  \[
    Q(v,w)\coloneqq \frac{1}{N}\big|\big\{ i\in[{\ell}n]\bigmid
    \pi\big((v,i)\big)\in\{w\}\times[{\ell}m]\big\}\big|.
  \]
  Then for every $v\in V$ we have
  \[
    \sum_{w\in W}Q(v,w)=\frac{{\ell}n}{N}=\frac{1}{m}.
  \]
  Furthermore, for every $w\in W$ we have
  \[
    \sum_{v\in V}Q(v,w)=\frac{{\ell}m}{N}=\frac{1}{n}.
  \]
  Thus $Q\in\DS^*(V,W)$. For all $v\in V,w\in W$, let $I(v,w)$ be the
  set of all $i\in[\ell n]$ such that
  $\pi\big((v,i)\big)\in\{w\}\times[\ell m]$. Then
  \begin{align*}
    \delta^{\textup{OT}}_{(1)}(G,H)&\le\sum_{v,v'\in V,w,w'\in
                       W}|A_G(v,v')-A_H(w,w')|Q(v,w)Q(v',w')\\
    &=\frac{1}{N^2}\sum_{v,v',w,w'
                      }\sum_{\substack{i\in I(v,w),\\i'\in I(v',w')}}\Big|A_{G'}\big((v,i),(v',i')\big)-A_{H'}\big(\pi(v,i),\pi(v',i')\big)\Big|\\
     &=\frac{1}{N^2}\sum_{(v,i),(v',i')\in
       V\times[{\ell}n]}\Big|A_{G'}\big((v,i),(v',i')\big)-A_{H'}\big(\pi(v,i),\pi(v',i')\big)\Big|\\
    &=\frac{1}{N^2}\delta_{(1)}(G',H').
  \end{align*}
  To prove the converse inequality, let $\epsilon>0$. We shall prove
  that
  \begin{equation}
    \label{eq:16}
    \delta_{(1)}^\odot(G,H)\le\delta_{(1)}^{\textup{OT}}(G,H)
    +\epsilon.
  \end{equation}
  Let
  $Q\in\DS^*(V,W)$ be a matrix with rational entries such that
  \[
    \sum_{v,v'\in V,w,w'\in
      W}\hspace{-3mm}\big|A_G(v,v')-A_H(w,w')\big|Q(v,w)Q(v',w')\le\delta_{(1)}^{\textup{OT}}(G,H)+\epsilon.
  \]
  Let $\ell\in\Nat$ be a common denominator for the entries of $Q$,
  and $q(v,w)$ such that $Q(v,w)=\frac{q(v,w)}{\ell}$. Note that
  $\sum_{w\in W}q(v,w)=\frac{\ell}{m}$ for all $v\in V$ and
  $\sum_{v\in V}q(v,w)=\frac{\ell}{n}$ for all $w\in W$. For all
  $v\in V$ we choose a partition $\big(I(v,w)\big)_{w\in W}$ of
  $[\ell n]$ such that $|I(v,w)|=mnq(v,w)$, and for all
  $w\in W$ we choose a partition $\big(J(v,w)\big)_{v\in V}$ of
  $[\ell m]$ such that $|J(v,w)|=mnq(v,w)$. For all $v,w$, we let
  $\pi_{vw}$ be a bijection from $I(v,w)$ to $J(v,w)$. 

  Let $G'\coloneqq G^{\odot\ell n}$ and $H'\coloneqq H^{\odot\ell m}$,
  and let $N\coloneqq\ell mn=|G'|=|H'|$. We let $\pi$ be the bijection
  from $V\times[\ell n]$ to $W\times[\ell m]$ defined by
  $\pi\big((v,i)\big)=\big(w,\pi_{vw}(i)\big)$ for all $v\in V,w\in
  W$, and $i\in I(v,w)$. Then, by essentially the same argument as in the proof
  of the other direction, we have
  \begin{align*}
    \frac{1}{N^2}\delta_{(1)}(G',H')
    &\le\frac{1}{N^2}\sum_{(v,i),(v',i')\in
      V\times[{\ell}n]}\Big|A_{G'}\big((v,i),(v',i')\big)-A_{H'}\big(\pi(v,i),\pi(v',i')\big)\Big|\\
    &=\frac{1}{N^2}\sum_{v,v'\in V,w,w'\in W
                      }\sum_{\substack{i\in I(v,w),\\i'\in
    I(v',w')}}\Big|A_{G'}\big((v,i),(v',i')\big)-A_{H'}\big(\pi(v,i),\pi(v',i')\big)\Big|\\
    &=      
    \sum_{v,v'\in V,w,w'\in
                       W}\big|A_G(v,v')-A_{H}(w,w')\big|\cdot\frac{|I(v,w)|}{N}\cdot\frac{|I(v',w')|}{N}\\
    &=      
    \sum_{v,v'\in V,w,w'\in
      W}|A_G(v,v')-A_H(w,w')|Q(v,w)Q(v',w')\\
    &\le\delta_{(1)}^{\textup{OT}}(G,H)+\epsilon.
\end{align*}             
This completes the proof of \eqref{eq:81}.
\end{proof}

We close this section by proving the inequality \eqref{eq:180} of
Theorem~\ref{thm:dot}. Our proof follows \cite{Pikhurko10}.
We need the
following lemma, which is a variant of Birkhoff's Theorem.

\begin{lemma}\label{lem:birkhoff}
  Let $A\in\Nat^{I\times I}$ such that there is an $m\in\Nat$ such
  that all row and column sums in $A$ are $m$. Then $A$ can be written
  as the sum of $m$ permutation matrices.
\end{lemma}
    
\begin{proof}
  Let $\ell\le m$ be maximum such that there are
  permutation matrices $P_1,\ldots,P_\ell$ such that
  \[
    \sum_{i=1}^\ell P_i\le A.
  \]
  (Inequalities between matrices are entrywise.)
  Then $B\coloneqq A-\sum_{i=1}^\ell P_i$ is a non-negative integer matrix with row and
  columns sums $m-\ell$. By the maximality of $\ell$, there is no
  permutation matrix $P$ such that $P\le B$. Thus by Hall's Marriage
  Theorem, there is a set $K\subseteq I$ such that less than $k\coloneqq|K|$
  columns $j$ have an entry $B(i,j)\ge 1$ with $i\in K$. Let $L$ be
  the set of all $j\in I$ such that $B(i,j)\ge 1$ for some $i\in
  K$. Then $|L|\le k-1$ and $B(i,j)=0$ for all $i\in K,j\not\in L$. Thus
  \begin{align*}
    k(m-\ell)=\sum_{i\in K,j\in I}B(i,j)=\sum_{i\in K,j\in L}B(i,j)\le
    \sum_{i\in I,j\in L}B(i,j)\le(k-1)(m-\ell).
  \end{align*}
  It follows that $(m-\ell)=0$ and thus $A=\sum_{i=1}^m P_i$.
\end{proof}

\begin{lemma}\label{lem:3to1}
  Let $\|\cdot\|$ be a permutation-invariant matrix norm.
  For finite sets $I,J$
  of the same cardinality, let
  $A\in\Real^{I\times I},B\in\Real^{J\times J}$ be symmetric matrices. Then for
  $\pi_1,\pi_2,\pi_3\in\Bij(I,J)$ we have
  \[
    \|A^{\pi_1,\pi_2}-B\|+\|A^{\pi_3,\pi_2}-B\|+\|A^{\pi_3,\pi_1}-B\|\ge
    \|A^{\pi_1}-B\|.
  \]
\end{lemma}

\begin{proof}
  Let $\iota$ denote the identity permutation (on either $I$ und
  $J$). Then by the permutation invariance of $\|\cdot\|$:
    \begin{align*}
      \|A^{\pi_1,\pi_2}-B\|&=\|A^{\pi_1}-B^{\iota,\pi_2^{-1}\pi_1}\|,\\
      \|A^{\pi_3,\pi_2}-B\|&=\|A^{\pi_3,\pi_1}-B^{\iota,\pi_2^{-1}\pi_1}\|
                             =\|B^{\iota,\pi_2^{-1}\pi_1}-A^{\pi_3,\pi_1}\|.
    \end{align*}
    Thus
    \begin{align*}
      &\|A^{\pi_1,\pi_2}-B\|+\|A^{\pi_3,\pi_2}-B\|+\|A^{\pi_3,\pi_1}-B\|\\
      &=\|A^{\pi_1}-B^{\iota,\pi_2^{-1}\pi_1}\|
        +\|B^{\iota,\pi_2^{-1}\pi_1}-A^{\pi_3,\pi_1}\|
        + \|A^{\pi_3,\pi_1}-B\|\\
      &\ge \|A^{\pi_1}-B^{\iota,\pi_2^{-1}\pi_1}
        +B^{\iota,\pi_2^{-1}\pi_1}-A^{\pi_3,\pi_1}
        +A^{\pi_3,\pi_1}-B\|\\
      &=\|A^{\pi_1}-B\|.
    \end{align*}
\end{proof}

\bigskip
\begin{proof}[Proof of \eqref{eq:180} (Theorem~\ref{thm:dot})]
  Let $G=(V,E_G)$ and $H=(W,E_H)$ be graphs of the same order
  $n\coloneqq|V|=|W|$, and let $\epsilon>0$. We shall prove that
  \begin{equation}
    \label{eq:9}
        \frac{1}{n^2}\delta_{(1)}(G,H)\le 3\delta_{(1)}^{\textup{OT}}(G,H)+\epsilon. 
  \end{equation}
  Then \eqref{eq:180} follows by \eqref{eq:81}.
  
Let
  $Q\in\DS^*(V,W)$ be a matrix with rational entries such that
  \[
    \sum_{v,v'\in V,w,w'\in
      W}\hspace{-3mm}\big|A_G(v,v')-A_H(w,w')\big|Q(v,w)Q(v',w')\le\delta_{(1)}^{\textup{OT}}(G,H)+\frac{\epsilon}{3}.
  \]
  Let $\ell\in\Nat$ be a common denominator for the entries of $Q$,
  and $q(v,w)$ such that $Q(v,w)=\frac{q(v,w)}{\ell}$. Observe that
  $\ell Q$ is a nonnegative integer matrix with row and column sums
  $m\coloneqq\frac{\ell}{n}$, and hence by Lemma~\ref{lem:birkhoff} there are permutation matrices
  $P_1,\ldots,P_m$ such that
  \[
    Q=\frac{1}{\ell}\sum_{i=1}^mP_i.
  \]
  For every $i\in[m]$, let $\pi_i\in\Bij(V,W)$ be the bijection
  corresponding to $P_i$, that is, $P_i(v,w)=1\iff\pi_i(v)=w$.
  Hence
  \begin{align*}
    &\sum_{v,v'\in V,w,w'\in
    W}\hspace{-3mm}\big|A_G(v,v')-A_H(w,w')\big|Q(v,w)Q(v',w')\\
    &=\frac{1}{\ell^2}\sum_{i,j\in[m]}\sum_{v,v'\in V,w,w'\in
    W}\hspace{-3mm}\big|A_G(v,v')-A_H(w,w')\big|P_i(v,w)P_j(v',w')\\
    &=\frac{1}{\ell^2}\sum_{i,j\in[m]}\sum_{v,v'\in
      V}\Big|A_G(v,v')-A_H\big(\pi_i(v),\pi_j(v')\big)\Big|\\
    &=\frac{1}{\ell^2}\sum_{i\in[m]}\sum_{v,v'\in
      V}\Big|A_G(v,v')-A_H\big(\pi_i(v),\pi_i(v')\big)\Big|\\
    &\hspace{4cm}+\frac{1}{\ell^2}\sum_{\substack{i,j\in[m],\\i\neq j}}\sum_{v,v'\in
    V}\Big|A_G(v,v')-A_H\big(\pi_i(v),\pi_j(v')\big)\Big|\\
    &=\frac{1}{\ell^2}\sum_{i=1}^m\|A_G^{\pi_i}-A_H\|_{(1)}+\frac{1}{\ell^2}\sum_{\substack{i,j\in[m],\\i\neq
      j}}\|A_G^{\pi_i,\pi_j}-A_H\|_{(1)}.
  \end{align*}
  For $i,j\in[m]$, we let
  $\Delta(i,j)\coloneqq \|A_G^{\pi_i,\pi_j}-A_H\|_{(1)}$. Then
  $\Delta(i,j)=\Delta(j,i)$ and, by Lemma~\ref{lem:3to1},
  $\Delta(i,j)+\Delta(j,k)+\Delta(k,i)\ge \Delta(i,i)$ for all
  $k$. Furthermore,
  $\Delta(i,i)=\|A_G^{\pi_i}-A_H\|_{(1)}\ge\delta_{(1)}(G,H)$.
  Thus
  \[
    \binom{m}{3}\delta_{(1)}(G,H)\le\sum_{\{i,j,k\}\in\binom{[m]}{3}}
    \big(\Delta(i,j)+\Delta(j,k)+\Delta(k,i)\big)
    =(m-2)\sum_{\{i,j\}\in\binom{[m]}{2}}\Delta(i,j),
  \]
  where the last equality holds because every 2-element subset of
  $[m]$ is contained in exactly $(m-2)$ 3-element subsets. It follows
  that
  \[
    \sum_{\substack{i,j\in[m]\\ i\neq
        j}}\Delta(i,j)=2\sum_{\{i,j\}\in\binom{[m]}{2}}\Delta(i,j)\ge\frac{m(m-1)}{3}\delta_{(1)}(G,H).
  \]
  Overall,
  \begin{align*}
    \delta_{(1)}^{\textup{OT}}(G,H)+\frac{\epsilon}{3}&\ge \frac{1}{\ell^2}\sum_{i=1}^m\Delta(i,i)+\frac{1}{\ell^2}\sum_{\substack{i,j\in[m],\\i\neq
    j}}\Delta(i,j)\\
    &\ge
      \frac{m}{\ell^2}\delta_{(1)}(G,H)+\frac{m(m-1)}{3\ell^2}\delta_{(1)}(G,H)\\
    &\ge\frac{m^2}{3\ell^2}\delta_{(1)}(G,H)\\
    &=\frac{1}{3n^2}\delta_{(1)}(G,H).
  \end{align*}
  This proves \eqref{eq:9}.
\end{proof}

\section{Declarative Distances}
Recall that the declarative view on graph similarity is that two
graphs are similar if they have similar properties. This idea leads to
the definition of graph metrics that are quite different from the
``operational'' metrics considered so far.

\subsection{Vector Embeddings}
A \emph{(graph) vector embedding} is an isomorphism invariant mapping
$\eta:\CG\to\mathbb V$ from graphs to some normed vector space
$\mathbb V$ over the reals,\footnote{We restrict our attention to real
  vector spaces here, though in principle we could also admit complex
  vector spaces.} whose norm we denote by $\|\cdot\|_{\mathbb V}$. \emph{Isomorphism invariant} means
that for isomorphic graphs $G,H$ we have $\eta(G)=\eta(H)$.  To be
able to develop a meaningful theory, we will usually assume that the
normed vector space $(\mathbb V,\|\cdot\|_{\mathbb V})$ satisfies
additional conditions such as being a Banach or Hilbert space, but
there is no need to worry about this here. Each
vector embedding $\eta:\CG\to\mathbb V$ induces a graph metric
$\delta_\eta$ defined by
\[
  \delta_\eta(G,H)\coloneqq\|\eta(G)-\eta(H)\|_{\mathbb V}.
\]
While this looks similar to the graph metrics defined by matrix norms,
note that the mapping $G\mapsto A_G$ is \emph{not} a vector embedding
in the sense defined here because the codomain $\Real^{V_G\times V_G}$
depends on the input graph $G$.

Note that, as opposed to the ``operational distances'' considered in
the previous section, the definition of $\delta_\eta$ does not
involve an alignment between the two graphs.

\begin{example}
  Consider the mapping $\eta:\CG\to\Real^4$ defined by
  \[
    \eta(G)\coloneqq
    \begin{pmatrix}
      |G|\\
      \text{number of edges of $G$}\\
      \text{number of triangles in $G$}\\
      \text{number of 4-cycles in $G$}
    \end{pmatrix}.
  \]
  (We equip $\Real^4$ with
  the usual Euclidean norm to turn it into a normed vector space.)
  If we want to use the induced metric $\delta_\eta$ to compare
  graphs of different orders, we may consider normalising the entries,
  for example, by defining $\tilde\eta: \CG\to\Real^4$ as
  \[
    \tilde\eta(G)\coloneqq
    \begin{pmatrix}
      1\\
      \frac{1}{|G|^2}\text{number of edges of $G$}\\
      \frac{1}{|G|^3}\text{number of triangles in $G$}\\
      \frac{1}{|G|^4}\text{number of 4-cycles in $G$}
    \end{pmatrix}.
  \]
  Of course here we can omit the first component and instead define
  $\tilde\eta$ as a mapping into the space $\Real^3$.
\end{example}

\begin{example}
  As a second example, consider the mapping $\CG\to\Real^k$
  mapping a graph to the $k$ largest eigenvalues of its adjacency
  matrix in decreasing order.
\end{example}

Generalising the approach taken in the previous examples, we can try
to collect features that we think are characteristic of the
properties of graphs that we are interested in for a specific application
scenario (think of the molecular graphs for synthetic fuels mentioned
as an example in the introduction) and define a vector embedding based on these features, which
gives us a tailored graph metric. We can also learn such a vector
embedding from data, for example using graph neural networks, which
leads to the area of \emph{graph representation learning} (see \cite{Hamilton20}).

\subsection{Graph Kernels}
For our theoretical considerations, nothing prevents us  from embedding
graphs into an infinite dimensional vector space $\mathbb
V$. Remarkably, this can even be useful in practice, where we are
mainly interested in the metric and often do not need the explicit
vector embedding. This leads to the idea of graph kernels.

A \emph{kernel function} for a set $\CX$ of objects is a binary
function $K:\CX\times\CX\to\Real$ that is symmetric, that is,
$K(x,y)=K(y,x)$ for all $x,y\in\CX$, and \emph{positive semidefinite},
that is, for all $n\ge 1$ and all $x_1,\ldots,x_n\in\CX$ the matrix
$A\in\Real^{n\times n}$ with entries $A(i,j):=K(x_i,x_j)$ is positive
semidefinite. Recall that a symmetric matrix $A\in\Real^{n\times n}$
is positive semidefinite if $\vec x^\top A\vec x\ge 0$ for all
$\vec x\in\Real^n$. It turns out that kernel functions are always
associated with vector embeddings into \emph{inner product spaces}. An
\emph{inner product space} is a (possibly infinite-dimensional) real
vector space $\mathbb I$ with a symmetric bilinear form
$\angles{\cdot,\cdot}_{\mathbb I}:\mathbb I\times \mathbb I\to\Real$
satisfying $\angles{\vec x,\vec x}_{\mathbb I}>0$ for all non-zero
$\vec x\in\mathbb I$.\footnote{If $\mathbb I$ is complete with respect
  to the metric defined by
  $\delta_{\mathbb I}(\vec x,\vec y)\coloneqq\angles{\vec x-\vec
    y,\vec x-\vec y}_{\mathbb I}$, that is, every Cauchy sequence
  converges, it is a \emph{Hilbert space}.} It is not
very hard to prove that a symmetric function $K:\CX\times\CX\to\Real$ is a kernel
function if and only if there is a vector embedding
$\eta:\CX\to\mathbb I$ of $\CX$ into some inner product space
$\mathbb I$, which
we can even take to be a Hilbert space, such
that $K$ is the mapping induced by $\angles{\cdot,\cdot}_{\mathbb I}$,
that is, $K(x,y)=\angles{\eta(x),\eta(y)}_{\mathbb I}$ for all $x,y\in\CX$ (see
\cite[Lemma~16.2]{ShalevB14} for a proof). For background on kernels
and specifically graph kernels, see \cite{KriegeJM20,ScholkopfS02}.

To give a meaningful and practically relevant example of a \emph{graph
  kernel}, that is, a kernel function on the class of graphs, this is
a good place to introduce the \emph{Weisfeiler-Leman algorithm}
\cite{Morgan65,WeisfeilerL68} (also see \cite{GroheKMS21,Kiefer20}), leading to \emph{Weisfeiler-Leman graph kernels}
\cite{ShervashidzeSLMB11}. The \emph{(1-dimensional) Weisfeiler-Leman algorithm}
(a.k.a.~\emph{colour refinement} or \emph{naive vertex
  classification}) iteratively colours the vertices of a graph $G$ as
follows. Initially, all vertices get the same colour. Then in each
iteration, two vertices $v,w$ get different colours if there is some
colour $c$ in the current colouring such that $v$ and $w$ have a
different number of neighbours of colour $c$. After at most $|G|-1$
iterations, the colouring becomes \emph{stable}, which means that any
two nodes $v,w$ of the same colour $c$ have the same number of
neighbours of any colour $d$. The WL algorithm \emph{distinguishes}
two graphs $G,H$ if there is a colour $c$ such that $G$ and $H$ have a
different number of vertices of colour $c$.

It is useful to think of the colours of 1-WL as rooted trees: the initial
colour of all vertices is just the 1-vertex tree, and the colour a
vertex $v$ receives in iteration $i+1$ is the tree obtained by
attaching the colours of all neighbours of $v$ in iteration $i$, which
are trees of height $i$, to a new root. This is illustrated in
Figure~\ref{fig:wl}.

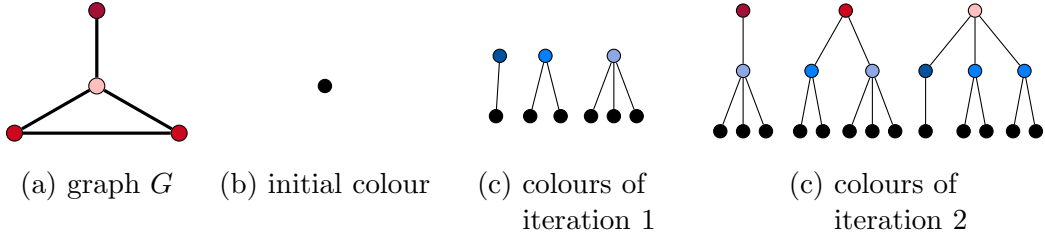
\begin{figure}
  \centering
  \begin{tikzpicture}[
    vertex/.style={circle,draw, minimum size=5pt,inner sep=0pt,fill=black},
    ]
    \begin{scope}
      \node[vertex,minimum size=6pt,fill=pink]
          (a) at (0,0) {};
          \node[vertex,minimum size=6pt,fill=bordeaux]
          (b) at (90:1cm) {};
          \node[vertex,minimum size=6pt,fill=rot]
          (c) at (210:1.25cm) {};
          \node[vertex,minimum size=6pt,fill=rot]
          (d) at (330:1.25cm) {};
          \draw[very thick] (a) edge (b) edge (c) edge (d) (c)
          edge (d);
          \path (0,-1) node[anchor=north] {(a) graph $G$};
        \end{scope}

        \begin{scope}[xshift=3cm]
          \node[vertex] (b) at (0,0) {};

          \path (0,-1) node[anchor=north] {(b) initial colour};
        \end{scope}

        \begin{scope}[xshift=6.5cm]
         \node[vertex,fill=blau] (c1) at (-1.2,0.4) {};
         \node[vertex] (c2) at (-1.25,-0.4) {};

         \draw (c1) edge (c2);

         \node[vertex,fill=mittelblau] (c3) at (-0.6,0.4) {};
         \node[vertex] (c4) at (-0.8,-0.4) {}; 
         \node[vertex] (c5) at (-0.4,-0.4) {}; 

         \draw (c3) edge (c4) edge (c5);
   
         \node[vertex,fill=hellblau] (c6) at (0.3,0.4) {};
         \node[vertex] (c7) at (0,-0.4) {}; 
         \node[vertex] (c8) at (0.3,-0.4) {}; 
         \node[vertex] (c9) at (0.6,-0.4) {};

         \draw (c6) edge (c7) edge (c8) edge (c9);
         
          \path (-0.2,-1) node[anchor=north] {(c)~\parbox[t]{2cm}{colours
              of\\iteration $1$}};
        \end{scope}

        \begin{scope}[xshift=10.5cm]
         \node[vertex,fill=bordeaux] (d1) at (-2,1) {};
         \node[vertex,fill=hellblau] (d2) at (-2,0.2) {};
         \node[vertex] (d3) at (-2.3,-0.6) {}; 
         \node[vertex] (d4) at (-2,-0.6) {}; 
         \node[vertex] (d5) at (-1.7,-0.6) {};

         \draw (d2) edge (d1) edge (d3) edge (d4) edge (d5);

         \node[vertex,fill=rot] (d6) at (-0.65,1) {};
         \node[vertex,fill=mittelblau] (d7) at (-1.1,0.2) {};
         \node[vertex] (d8) at (-1.25,-0.6) {}; 
         \node[vertex] (d9) at (-0.95,-0.6) {}; 
         \node[vertex,fill=hellblau] (d10) at (-0.3,0.2) {};
         \node[vertex] (d11) at (-0.6,-0.6) {}; 
         \node[vertex] (d12) at (-0.3,-0.6) {}; 
         \node[vertex] (d13) at (0,-0.6) {}; 

         \draw (d7) edge (d6) edge (d8) edge (d9) (d10) edge (d6) edge
         (d11) edge (d12) edge (d13);

          \node[vertex,fill=pink] (d14) at (1.05,1) {};
          \node[vertex,fill=blau] (d15) at (0.4,0.2) {};
          \node[vertex] (d16) at (0.4,-0.6) {}; 
          \node[vertex,fill=mittelblau] (d17) at (1.05,0.2) {};
          \node[vertex] (d18) at (0.9,-0.6) {}; 
          \node[vertex] (d19) at (1.2,-0.6) {}; 
          \node[vertex,fill=mittelblau] (d20) at (1.7,0.2) {};
          \node[vertex] (d21) at (1.55,-0.6) {}; 
          \node[vertex] (d22) at (1.85,-0.6) {}; 

          \draw (d15) edge (d14) edge (d16) (d17) edge (d14) edge
          (d18) edge (d19) (d20) edge (d14) edge (d21) edge (d22);

          \path (-0.1,-1) node[anchor=north] {(c)~\parbox[t]{2cm}{colours
              of\\iteration $2$}};
        \end{scope}

  \end{tikzpicture}
  \caption{Weisfeiler-Leman colours viewed as trees.}
  \label{fig:wl}
\end{figure}

For every $i\in\Nat$, let $\CC_i$ be the set of all colours that can
be obtained by running the WL-algorithm on some graph for
$i$ iterations. If we view colours as trees, $\CC_i$ is the set of all
rooted trees where all leaves have distance exactly $i$ from the
root. Furthermore, let $\CC\coloneqq\bigcup_{i\in\Nat}\CC_i$ be the
set of all colours. With every graph $G$ we associate a vector
$\wl G\in\Real^{\CC}$, where for each $c\in\CC$ the entry
$\wl G(c)$ is the number of vertices of $G$ of colour $c$. Note
that the WL algorithm distinguishes graphs $G$ and $H$ if and only if
$\wl G\neq\wl H$.

Observe that for every graph $G$ the vector $\wl G$ belongs to the
subspace $\VWL$ of $\Real^{\CC}$ consisting of all $\vec
w$ such that for some $n\in\Nat$ 
we have $\sum_{c\in\CC_i}\vec w(c)\le n$ for all $i\in\Nat$. If $\vec
w=\wl G$ then we take $n\coloneqq|G|$.
We can define an inner product $\anglesWL{\cdot,\cdot}$ on
$\VWL$ by
\begin{equation}
  \label{eq:17}
    \anglesWL{\vec w,\vec
    w'}\coloneqq\sum_{i\in\Nat}\frac{1}{2^i}\sum_{c\in\CC_i}\vec
  w(c)\vec w'(c).
\end{equation}
Note that the normalisation factor $1/2^i$ is quite arbitrary and only
introduced to make the series converge.
We define the \emph{Weisfeiler-Leman graph kernel}
$\KWL:\CG\times\CG\to\Real$ by
\begin{equation}
  \label{eq:20}
  \KWL(G,H)\coloneqq\anglesWL{\wl G,\wl H}. 
\end{equation}
Here, the beauty of the kernel idea becomes clear: to compute
$\KWL(G,H)$ we never have to worry about the infinite vectors
$\wl G,\wl H$, but we only have to run the WL algorithm on the
disjoint union of the graphs $G$ and $H$ and compute the stable
colouring. This can be done very efficiently; the stable colouring of
a graph of order $n$ with $m$ edges can be computed in time
$O(m\log n)$ \cite{CardonC82}. In practice, it turns out to be best to only
look at the first few iterations of the WL algorithm. Then we can also
omit the factor $1/2^i$ in the definition of the inner product and
define an inner product by
$\sum_{i\le\ell}\sum_{c\in\CC_i}\vec w(c)\vec w'(c)$, where $\ell$ is
the number of iterations. Empirically, $\ell=5$ iterations seem to
work well \cite{ShervashidzeSLMB11}.

The kernel $\KWL$ can be viewed as a graph similarity measure. It
can be turned into a graph metric $\deltaWL$ defined by
\begin{equation}
  \label{eq:19}
  \deltaWL(G,H)\coloneqq \sqrt{\KWL(G,G)-2\KWL(G,H)+\KWL(H,H)}; 
\end{equation}
note that this is exactly $\|\wl G-\wl H\|_{\textup{WL}}=\sqrt{\angles{\wl G-\wl H, \wl G-\wl H}_{\textup{WL}}}$.

Besides the Weisfeiler-Leman graph kernel, there are many other graph
kernels, for example, based on comparing random walks in the graphs \cite{garflawro03} or
based on counting small subgraphs \cite{shevispet+09}. The
homomorphism embeddings we shall discuss in the following section can be seen
as a variant of the latter.

\subsection{Homomorphism Counts and Densities}
A \emph{homomorphism} from a graph $F=(U,E_F)$ to a graph $G=(V,E_G)$
is a mapping $h:U\to V$ such that $uu'\in E_F$ implies
$h(u)h(u')\in E_G$. We denote the number of homomorphisms from $F$ to
$G$ by $\hom(F,G)$. Homomorphism numbers give us useful information
about a graph. For example, for every graph $G=(V,E)$ we have $\hom(\tikz{\node[fill,circle,inner sep=0pt,minimum
  size=1.5mm] {};},G)=|V|$ and
$\hom(\tikz{
  \node[fill,circle,inner sep=0pt,minimum size=1.5mm] (v1) {};
  \node[fill,circle,inner sep=0pt,minimum size=1.5mm] (v2) at (0.5,0) {};
\draw[thick] (v1) edge (v2);
},G)=2|E|$. Thus it makes sense to use homomorphism counts as features
for vector embeddings.

However, to define such vector embeddings, we prefer to work with
normalised homomorphism counts. Let
$k\coloneqq|F|$ and $n\coloneqq|G|$, and assume that $n\ge 1$. The \emph{homomorphism density} of $F$ in
$G$ is defined to be
\[
  \hd(F,G)\coloneqq\frac{1}{n^k}\hom(F,G)=\Pr_{h\in V^U}(h\text{ is a
    homomorphism from $F$ to $G$})
\]
where the probability ranges over $h$ picked uniformly at random from
the set of all mappings from $U$ to $V$. Now for any class $\CF$ of
graphs we can define a vector embedding $\eta_{\CF}:\CG\to\Real^{\CF}$ by
$\eta_{\CF}(G)(F)\coloneqq\hd(F,G)$. Note that $\eta_{\CF}(G)$
is in the subspace $\VF\subseteq\Real^{\CF}$ consisting of all bounded
vectors, that is, all $\vec v\in\Real^{\CF}$ such that for some
$k\in\Nat$ it holds that $|\vec v(F)|\le k$ for all $F\in\CF$. We
define an inner product on $\VF$ by
\[
  \anglesF{\vec v,\vec
    w}\coloneqq\sum_{k\in\Nat}\frac{1}{2^k|\CF_k|}\sum_{F\in\CF_k}\vec
  v(F)\vec w(F).
\]
This inner product gives us a graph kernel
$K_\CF(G,H)\coloneqq\anglesF{\eta_{\CF}(G),\eta_{\CF}(H)}$, and from
$K_\CF$ we can define a metric $\deltaF$ as in \eqref{eq:19}.
We obtain
\begin{equation}
  \label{eq:32}
    \deltaF(G,H)=\sqrt{\sum_{k\in\Nat}\frac{1}{2^k|\CF_k|}\sum_{F\in\CF_k}\big(\hd(F,G)-\hd(F,H)\big)^2}.
\end{equation}
With the interpretation of $\hd(F,G)$
as the probability that a randomly chosen mapping from the vertex set
of $F$ to the vertex set of $G$ is a homomorphism, it is easy to see that
$\hd(F,G)=\hd(F,G^{\odot k})$ for all $k$. Hence
$\deltaF(G,G^{\odot k})=0$ for all $\CF,G,k$.

Before studying the metrics $\deltaF$, we can try
to understand the equivalence relation induced by the classes of
graphs of mutual distance $0$. Two graphs $G,H$ are \emph{homomorphism
  indistinguishable over $\CF$} if $\hom(F,G)=\hom(F,H)$ for all
$F\in\CF$. To relate this to the metric $\deltaF$, we note that for
all graphs $G,H$ and all classes $\CF$ of graphs we have:
\begin{equation}
  \label{eq:21}
  \deltaF(G,H)=0\iff\parbox[t]{8cm}{there are $k,\ell\in\PNat$ such that
    $G^{\odot k}$ and $H^{\odot\ell}$ are homomorphism
    indistinguishable over $\CF$.}
\end{equation}
The backward direction of this equivalence follows from the fact that
$\deltaF(G,G^{\odot k})=0$ for all $G,k$.

Let us prove the forward direction. Assume $\deltaF(G,H)=0$.  If
$|G|=|H|$ then for all graphs $F$ it holds that
$\hom(F,G)=\hom(F,H)\iff\hd(F,G)=\hd(F,H)$. Thus if $\deltaF(G,H)=0$
then $G,H$ are homomorphism indistinguishable over $\CF$.

Otherwise, suppose $m\coloneqq |G|\neq|H|\eqqcolon n$. Observe that
\[
    0\le\deltaF(G^{\odot n},H^{\odot m})\le\deltaF(G,G^{\odot
      n})+\deltaF(G,H)+\deltaF(H,H^{\odot m})=0
  \]
  Thus $\deltaF(G^{\odot n},H^{\odot m})=0$. As $|G^{\odot
    n}|=|H^{\odot m}|$, this implies that $G^{\odot
    n}$ and $H^{\odot m}$ are homomorphism indistinguishable over $\CF$.

\begin{theorem}[Lovász~\cite{Lovasz67}]
  Two graphs $G,H$ are isomorphic if and only if they are homomorphism
  indistinguishable over the class $\CG$ of all graphs.
\end{theorem}

Then from \eqref{eq:21}, \eqref{eq:18}, \eqref{eq:19} and the fact
that $\delta_{(1)}^{\odot}(G,H)\le \delta_{(1)}(G,H)$,
$\delta_{\square}^{\odot}(G,H)\le \delta_{\square}(G,H)$ for all graph
$G,H$ of the same order we obtain the following.

\begin{corollary}
  For all graph $G,H$ we have
  \[
    \delta_{\CG}(G,H)=0\iff\delta_{(1)}^{\odot}(G,H)=0\iff\delta_{\square}^{\odot}(G,H)=0.
  \]
\end{corollary}

Many other natural equivalence relations on graphs (besides
isomorphism) can be characterised in terms of homomorphism
indistinguishability \cite{Dvorak10,DellGR18,MancinskaR20,gro20a,GroheRS21}. For example, it is not difficult to
see that two graphs are homomorphism indistinguishable over the class
of cycles if and only if they are co-spectral, that is, their adjacency
matrices have the same eigenvalues with the same multiplicities.
The following theorem relates homomorphism indistinguishability over the
class $\CT$ of all trees to the fractional relaxations of our graph metrics. 

\begin{theorem}[Dvorák~\cite{Dvorak10}, Tinhofer~\cite{Tinhofer91}]
 For all graphs $G,H$ the
  following are equivalent.
  \begin{enumerate}
  \item $G,H$ are homomorphism indistinguishable over the class of all
    trees.
  \item $G,H$ are \emph{fractionally isomorphic}, that is, there is a
    doubly stochastic matrix $Q$ such that $A_GQ=Q A_H$.
  \item The WL-algorithm does not distinguish $G,H$ and thus $\deltaWL(G,H)=0$. 
  \end{enumerate}
\end{theorem}

\begin{corollary}
  For all graph $G,H$ we have
  \[
    \delta_{\CT}(G,H)=0\iff\delta_{(1)}^*(G,H)=0\iff\delta_{\square}^*(G,H)=0.
  \]
\end{corollary}

The previous two theorems are very interesting because they connect
operational distances like $\delta_{\square}^{(*)}$ and $\delta_{(1)}^{(*)}$ to
declarative distances like $\delta_{\CF}$ and $\deltaWL$. However, the
connection is weak, as it only connects the equivalence relations induced
by the respective metrics being $0$. Ideally, we would like to have
tight quantitive
bounds between the different metrics. A qualitative connection between
two metrics that goes significantly beyond just establishing that they
define the same equivalence relation is to prove that they define the
same topology. We say that two pseudo-metrics $\delta,\delta'$ on a
space $X$ 
\emph{define the same topology} if and only if for every $\epsilon>0$
there is an $\epsilon'>0$ such that
$\delta(x,y)\le\epsilon'\implies\delta'(x,y)\le\epsilon$ and
$\delta'(x,y)\le\epsilon'\implies\delta(x,y)\le\epsilon$. In other
words, the metrics define the same topology if the identity mapping is
a continuous mapping between the two metric spaces in both directions.

\begin{theorem}[Borgs et al.~\cite{BorgsCLSV08}, Böker~\cite{Boker21}] 
  \begin{enumerate}
  \item
    $\delta^{\odot}_\square$ and $\delta_{\CG}$ define the same
    topology.
  \item $\delta^{*}_\square$ and $\delta_{\CT}$ define the same
    topology.
  \end{enumerate}
\end{theorem}

Assertion (2) does not extend to the metric $\deltaWL$, but Böker et
al.\ \cite{BokerLHVM23} introduced a more refined metric based on
the Weisfeiler-Leman algorithm and optimal transport and showed that
it also defines the same topology as $\delta_\CT$. There is also
extension of the theorem from trees to graphs of bounded tree width
\cite{Boker23}.

\subsection{Subgraph Densities and Sampling Distance}
Most of the graph metrics considered so far suffer from the fact that
they are hard to compute (see Section~\ref{sec:complexity}). At least
intuitively, an easy way of comparing to graphs $G,H$ is by sampling
small subgraphs, say of order $k\le|G|,|H|$ and comparing the
distributions on graphs of order $k$ obtained this way when sampling
in $G$ and $H$, respectively. This is formalised in the \emph{sampling
 distance} we consider in this section. All results of this section
can be found, at least implicitly, in \cite{Lovasz12}.

Let $F=(U,E_F)$ and $G=(V,E_G)$ be graphs of orders $k\coloneqq|F|$, $n\coloneqq|G|$.
An \emph{embedding} of $F$ into $G$
is an injective homomorphism from $F$ to $G$. An embedding $h$ is
\emph{strong} if for all $u,u'\in U$ it holds that
$uu'\in E_F\iff h(u)h(u')\in E_G$. Observe that there is an embedding
of $F$ into $G$ if and only if $G$ has a subgraph isomorphism to $F$,
and there is a strong embedding of $F$ into $G$ if and only if $G$ has
an induced subgraph isomorphic to $F$. Similarly to the homomorphism
numbers and homomorphism densities, we define (induced) subgraph
numbers and densities: 
\begin{align*}
  \emb(F,G)&\coloneqq\big|\big\{h:U\to V\bigmid h\text{ is an
             embedding from $F$ into $G$}\big\}\big|,\\
  \semb(F,G)&\coloneqq\big|\big\{h:U\to V\bigmid h\text{ is a strong
             embedding from $F$ into $G$}\big\}\big|,\\
  \ed(F,G)&\coloneqq \Pr_{h\in\Inj(U,V)}(h\text{ is an embedding of $F$ into $G$}),\\
  \sd(F,G)&\coloneqq \Pr_{h\in\Inj(U,V)}(h\text{ is a strong embedding of $F$ into $G$}).
\end{align*}
Note that  $\ed(F,G)$ and $\sd(F,G)$ are only defined if $k\le
n$. We set $\ed(F,G)\coloneqq\sd(F,G)\coloneqq0$ if $k>n$.

Clearly,
\[
  \sd(F,G)\le\ed(F,G).
\]
If we want to compare $\ed(F,G)$ and $\sd(F,G)$ with $\hd(F,G)$, we
have to deal with non-injective mappings. Observe that 
\[
  \Pr_{h\in V^U}(h\text{ is not
    injective})\le\binom{k}{2}\frac{1}{n}.
\]
which is below $\epsilon>0$ for $n\ge\frac{k(k-1)}{2\epsilon}$.
Thus for $\epsilon>0$ and $n\ge\frac{k(k-1)}{2\epsilon}$ we have
\begin{equation}
  \label{eq:27}
  \sd(F,G)\le\ed(F,G)\le\hd(F,G)+\Pr_{h\in V^U}(h\text{ is not
    injective})=\hd(F,G)+\epsilon.
\end{equation}

\begin{example}
  Assume $2\le k\le n$.
  \begin{enumerate}
  \item
  Let $F=K_k$ and
  $G=K_n$ be complete graphs. Then
  \[
    \hd(F,G)=\Pr_{h\in [n]^{[k]}}(h\text{ is
      injective})<1=\ed(F,G)=\sd(F,G).
  \]
  \item Let $F=([k],\emptyset)$ be an edgeless graph and $G$ a
    complete graph. Then $\hd(F,G)=\ed(F,G)=1$ and $\sd(F,G)=0$.
  \end{enumerate}
\end{example}

The previous example shows that we cannot bound $\ed(F,G)$ or $\hd(F,G)$ in terms of
$\sd(F,G)$.
However, we can express both homomorphism numbers and embedding
numbers in terms of strong embedding numbers (see, e.g., \cite[Section~5.2.3]{Lovasz12}). Consider graphs
$F=(U,E_F)$ and $G=(V,E_G)$. Then 
\begin{equation}
  \label{eq:23}
    \emb(F,G)=\sum_{E'\subseteq\binom{U}{2}\text{ with }E_F\subseteq
    E'}\semb\big((U,E'),G\big).
\end{equation}
To express homomorphism counts, for a graph $F$ and a partition $P$ of $V_F$ we define the
\emph{quotient graph} $F/P$ to be the graph with vertex set $P$ and
edges $pq$ for all $p,q\in P$ such that there is an edge $vw\in E_F$
with $v\in p, w\in q$. The graph $F/P$ is not necessarily a simple
graph; it may have loops.
We have
\begin{align}
  \label{eq:24}
  \hom(F,G)&=\sum_{P\text{ partition of }U}\emb(F/P,G)\\
  \label{eq:25}
  &=\sum_{P\text{ partition of }U}\sum_{E'\subseteq\binom{P}{2}\text{
    with }E_{F/P}\subseteq E'}\semb\big((P,E'),G\big).
\end{align}
Note that if $F/P$ has loops, then $\emb(F/P,G)=0$. This means that we
can restrict the sum to partitions $P$ into independent sets of $F$,
that is, for all parts
$p\in P$ and all $v,w\in p$ it holds that $vw\not\in E_F$.

Conversely, we can express $\emb$ and $\semb$ in terms of $\hom$:
\begin{equation}
  \label{eq:30}
  \emb(F,G)=\sum_{P\text{ partition of }U}\mu_P\hom(F/P,G),
\end{equation}
where $\mu_P=(-1)^{n-|P|}\prod_{p\in P}\big(|p|-1\big)!$, and
\begin{equation}
  \label{eq:31}
  \semb(F,G)=\sum_{E'\subseteq\binom{U}{2}\text{ with }E_F\subseteq
    E'}(-1)^{|E'|-|E_F|}\emb\big((U,E'),G\big).
\end{equation}

\begin{lemma}\label{lem:sd-ed-hd}
  Let $G,H$ be graphs and $\epsilon\ge 0$, $k\in\Nat$.
  \begin{enumerate}
  \item If for all  $F\in\CG_k$ it
  holds that $\big|\sd(F,G)-\sd(F,H)\big|\le\epsilon$, then for all $F\in\CG_k$ it
  holds that $\big|\ed(F,G)-\ed(F,H)\big|\le 2^{O(k^2)}\epsilon$. 
  \item If for all  $F\in\CG_{\le k}$ it
  holds that $\big|\sd(F,G)-\sd(F,H)\big|\le\epsilon$, then for all $F\in\CG_k$ it
  holds that $\big|\hd(F,G)-\hd(F,H)\big|\le 2^{O(k^2)}\epsilon$. 
  \item If for all  $F\in\CG_{\le k}$ it
  holds that $\big|\hd(F,G)-\hd(F,H)\big|\le\epsilon$, then for all $F\in\CG_k$ it
  holds that $\big|\ed(F,G)-\ed(F,H)\big|\le 2^{O(k^2)}\epsilon$ and $\big|\sd(F,G)-\sd(F,H)\big|\le 2^{O(k^2)}\epsilon$. 
  \end{enumerate}
\end{lemma}

\begin{proof}
  Assertions (1) and (2) follow immediately from \eqref{eq:23} and
  \eqref{eq:25}, because $|\CG_{\le k}|=2^{O(k^2)}$ and the number of
  partitions of a $k$-element set is also $2^{O(k^2)}$. For assertion
  (3), we apply \eqref{eq:30} and \eqref{eq:31}, additionally
  observing that the absolute value of the coefficients $\mu_P$ in
  \eqref{eq:30} is bounded by $2^{O(k^2)}$.
\end{proof}

Lovász~\cite{Lovasz12} defines the \emph{sampling distance} between
two graphs based on the total variation distance of the distributions
of induced subgraphs sampled from the two graphs.

The \emph{total variation distance} between two probability
distributions $p,q$ on the same space $\Omega$, which for simplicity we
assume to be finite here, is
\[
  \dist_{\textup{TV}}(p,q)\coloneqq\frac{1}{2}\sum_{\omega\in\Omega}|p(\omega)-q(\omega)|.
\]
Note that 
  $\dist_{\textup{TV}}(p,q)=\frac{1}{2}\|p-q\|_1$ if we regard $p,q$ as
  vectors in $\Real^\Omega$. It is also not hard to see
(cf.~\cite[Section~12.1]{MitzenmacherU17}) that
\[
  \dist_{\textup{TV}}(p,q)=\sup_{A\subseteq\Omega}|p(A)-q(A)|,
\]
where $p(A)\coloneqq\sum_{\omega\in A}p(A)$ and similarly for $q$.

Let $k\in\Nat$. Then
with every graph $G$ of order $|G|\ge k$
we associate a probability distribution $p_{G,k}$ on $\CG_k$ defined by
\[
  p_{G,k}(F)\coloneqq\sd(F,G)=\Pr_{h\in\Inj(U,V)}\big(h:F\cong G[h(U)]\big).
\]
Thus $p_{G,k}(F)$ is the probability that the induced subgraph on
a sample of $k$ vertices without repetitions is $F$.
To avoid cumbersome case distinctions, for graphs $G$ of order
$|G|<k$ we define $p_{G,k}$ by letting
$p_{G,k}\big(([k],\emptyset)\big)\coloneqq1$
and $p_{G,k}\big(F\big)\coloneqq0$ for all graphs with at least one
edge.

Then we define the
\emph{sampling distance} between graphs $G$ and $H$ to be
\[
  \dsamp(G,H)\coloneqq\sum_{k\in\PNat}\frac{1}{2^k}\dist_{\textup{TV}}(p_{G,k},p_{H,k})=\sum_{k\in\PNat}\frac{1}{2^{k+1}}\sum_{F\in\CG_k}\big|p_{G,k}(F)-p_{H,k}(F)\big|.
\]
Even though the normalising factor $\frac{1}{2^k}$ is not necessary to
guarantee convergence, because $\delta_{TV}(p_{G,k},p_{H,k})=0$ for
all $k>\max\{|G|,|H|\}$, it is useful to make sure that the sum is
dominated by the values at small $k$.

\begin{example}
  Let $G\coloneqq K_m$ and $H\coloneqq K_n$, where $m<n$. Observe
  that for $m<k\le n$ we have $p_{G,k}(K_k)=0$ and $p_{H,k}(K_k)=1$
  and, in fact, $\delta_{\textup{TV}}(p_{G,k},p_{H,k})=1$. However, we
  think of the complete graphs $K_m$ and $K_N$ to be similar for large
  $m,n$, even if $m,n$ are far apart.
\end{example}

\begin{theorem}
  $\delta_\CG$ and $\dsamp$ define the same topology.
\end{theorem}

\begin{proof}
  Let $G=(V,E_G)$, $H=(W,E_H)$ be graphs and $m\coloneqq|G|$,
  $n\coloneqq|H|$. Let $0<\epsilon\le 1$, and let
  $k_\epsilon\coloneqq\ceil{\log\frac{2}{\epsilon}}$. Without loss
  of generality, we may assume that $m,n>k_\epsilon$.

  Assume first that
  \begin{equation}
    \label{eq:29}
    \delta_{\CG}(G,H)\le\sqrt{\frac{\epsilon}{2^{ck_\epsilon^2}}}
  \end{equation}
  for a sufficiently large constant $c$ that we will determine later. We
  shall prove that $\dsamp(G,H)\le\epsilon$. Since
  $\dist_{\textup{TV}}(p_{G,k},p_{H,k})\le 1$ for all $k$, we have
  \[
    \sum_{k>k_\epsilon}\frac{1}{2^k}\dist_{\textup{TV}}(p_{G,k},p_{H,k})\le
    \sum_{k>k_\epsilon}\frac{1}{2^k}=\frac{1}{2^{k_\epsilon}}\le\frac{\epsilon}{2}.
  \]
  Thus it suffices to prove that for every $k\le k_\epsilon$ and every
  $F\in\CG_k$ we have
  \begin{equation}
    \label{eq:28}
    \big|\sd(F,G)-\sd(F,H)\big|=\big|p_{G,k}(F)-p_{H,k}(F)\big|\le\frac{\epsilon}{|\CG_k|}.
  \end{equation}
  Indeed, \eqref{eq:28} implies
  $\dist_{\textup{TV}}(p_{G,k},p_{H,k})\le\frac{\epsilon}{2}$ and thus
  \[
    \sum_{1\le k\le k_\epsilon}\frac{1}{2^k}\dist_{\textup{TV}}(p_{G,k},p_{H,k})\le
    \frac{\epsilon}{2}\sum_{k\ge 1}\frac{1}{2^k}=
    \frac{\epsilon}{2}
  \]
  So let $1\le k\le k_{\epsilon}$ and $F\in\CG_k$. Without loss of
  generality, we assume that $\sd(F,G)\ge\sd(F,H)$. Since $m\ge
  \frac{1}{\epsilon}2^{k_\epsilon^2}\ge
  \frac{k(k-1)|\CG_k|}{\epsilon}$ we have
  \[
    \Pr_{h\in V^{[k]}}\big(h\text{ is not
      injective}\big)\le\frac{\epsilon}{2|\CG_k|}
  \]
  and thus
  \[
    \sd(F,G)\le \hd(F,G)+\frac{\epsilon}{2|\CG_k|}.
  \]
  By Lemma~\ref{lem:sd-ed-hd}(3), to prove \eqref{eq:28} it suffices
  to prove that for all $k'\le k$ and all $F'\in\CG_{k'}$ it holds
  that
  \[
    \big|\hd(F',G)-\hd(F',H)\big|\le \frac{\epsilon}{2^{c'k^2}}
  \]
  for a suitable constant
  $c'$. Suppose for contradiction that for some $k'\le
  k$ and $F'\in\CG_{k'}$ we have
  \[
    \big|\hd(F',G)-\hd(F',H)\big|> \frac{\epsilon}{2^{c'k^2}}
  \]
  Then
  \begin{align*}
    \delta_{\CG}(G,H)^2&=\sum_{k\in\Nat}\frac{1}{2^k|\CG_k|}\sum_{F\in\CG_k}\big(\hd(F,G)-\hd(F,H)\big)^2\\
    &\ge \frac{1}{2^{(k')^2}}\big|\hd(F',G)-\hd(F',H)\big|> \frac{\epsilon}{2^{(c'+1)k^2}},
  \end{align*}
  which contradicts \eqref{eq:29} for $c=c'+1$.

  \medskip
  The proof of the converse direction is very similar. We assume that
  $
    \dsamp(G,H)\le\frac{\epsilon}{2^{ck_\epsilon^2}}
  $
  for a sufficiently large $c$. Then we show that
  $\delta_{\CG}(G,H)\le\epsilon$ using Lemma~\ref{lem:sd-ed-hd}(2).
\end{proof}

\subsection{Logic and Games}
Another natural way of defining declarative similarities between
graphs is through logical equivalence. Here, we briefly discuss some
observations about this approach linking it to what we discussed
before. Related work can, for example, be found in \cite{bacbaclarmar19,BauerBKR08,larfahthr11,YingW00}.

For a
logic $\LL$,\footnote{For our purposes, it is sufficient to think of a
  logic $\LL$ as a set of \emph{sentences} together with a binary
  \emph{satisfaction relation} $\models\subseteq\CG\times\LL$ that is
  isomorphism invariant in the sense that for isomorphic graphs $G,H$
  and $\phi\in\LL$
  we have $G\models\phi\iff H\models\phi$. If $G\models\phi$ we say
  that $G$ \emph{satisfies} $\phi$ (see \cite{ebb85} for background on
  abstract logics).} we define the \emph{$\LL$-equivalence relation} $\equiv_\LL$,
where two graphs are equivalent if and only if they satisfy the same
sentences of the logic $\LL$.  For first-order logic $\textsf{FO}$,
this just yields the isomorphism relation, but for other logics, we
obtain interesting new relations such as \emph{bisimilarity} for
modal logic $\textsf{ML}$ or indistinguishability by the
Weisfeiler-Leman algorithm for the 2-variable fragment $\LC^2$ of
first-order logic with counting.

We can use the equivalence relation $\equiv_\LL$ to define a trivial
graph metric $\delta_\LL$ by letting $\delta_\LL(G,H)\coloneqq0$ if
$G\equiv_{\LL} H$ and $\delta_\LL(G,H)\coloneqq1$ otherwise. Note that
$\delta_{\FO}$ is just the isomorphism distance
$\delta_{\cong}$. To obtain ``real'' quantitative metrics from logical
equivalence, we can stratify our logic $\LL$ into a family
$(\LL_k)_{k\ge 1}$ such that $\LL_k\subseteq\LL_{k+1}$ for all $k$
and $\bigcup_k\LL_k=\LL$. Then, we can define a metric
$\delta_{(\LL_k)}$ by $\delta_{(\LL_k)}(G,H)\coloneqq 0$ if
$G\equiv_{\LL} H$ and $\delta_{(\LL_k)}(G,H)\coloneqq\frac{1}{k}$ for
the least $k$ such that $G\not\equiv_{\LL_k}H$.
For example, for first-order logic, we can
consider the stratifications by quantifier rank or by the number of
variables. In fact, the metric $\delta_{(\LL_k)}$ is an
\emph{ultrametric} with a stronger form triangle equation:
$\delta_{(\LL_k)}(F,H)\le\max\big\{\delta_{(\LL_k)}(F,G),
\delta_{(\LL_k)}(G,H)\big\}$ for all graphs $F,G,H$.

Often, equivalence in a logic can be characterised in terms of
so-called Ehrenfeucht-\Fraisse\ games, with parameters like ``number of
rounds'' or ``number of pebbles'' naturally corresponding to
stratifications of the logic by quantifier rank or number of
variables. This way, the games naturally align with the metrics. In fact, since the games establish some kind of
correspondence between the two graphs, we may view them as
operational counterparts of our declarative logical distances.

Another generic way of defining a graph metric based on a logic $\LL$
is by using the embedding approach: we define a graph embedding
$\eta_{\LL}:\CG\to\Real^{\LL}$ by $\eta_{\LL}(G)(\phi)\coloneqq 1$ if
$G$ satisfies $\phi\in\LL$ and $\eta_{\LL}(G)(\phi)\coloneqq 0$
otherwise. Then we can define an inner product on the subspace
$\mathbb V_\LL$ of all bounded vectors in $\Real^{\LL}$, obtain a
graph kernel, and use it to define a metric.

It is an interesting research project to explore such logically
defined metrics and understand how they relate to the other metrics we defined.

\section{Complexity}
\label{sec:complexity}
One may view similarity as ``approximate isomorphism''
(cf.~\cite{ArvindKKV12}), which at first sight may suggest that it is
an easier problem than the notorious graph isomorphism problem. But
while graph isomorphism is in quasi-polynomial time \cite{Babai16} and
for various other reasons unlikely to be NP-complete, similarity
with respect to most metrics turns out to be a much harder
problem.

For the operational distances, intuitively, we have to optimise over
the unwieldy set of all bijections to find a good alignment
between two graphs, while losing the group structure that gives us a
handle on the isomorphism problem.
To illustrate the hardness, let us argue that the problem of deciding
whether the edit distance between graphs $G,H$ is at most $k$ is
NP-complete. We reduce from the Hamiltonian-Cycle problem: a graph $G$
with $n$ vertices and $m$ edges has a Hamiltonian cycle if and only if
the edit distance between $G$ and a cycle of length $n$ is at most
$m-n$. Using more refined arguments, it can be shown that edit
distance is even hard if both input graphs are trees \cite{GroheRW18},
and it is also hard to approximate~\cite{ArvindKKV12}. Similar results
hold for most other graph metrics based on matrix
norms~\cite{GervensG22}.  Not much is known about the complexity of
the graph metrics based on optimal transport (e.g.\
$\delta_{(1)}^{\textup{OT}}$) or the blow-operation (e.g.\
$\delta_\square^{\odot}$), though I conjecture that they are also hard
to compute. Only the convex relaxations (e.g.\ $\delta_{(1)}^*$) can be computed
efficiently using gradient descent methods. There is an interesting
connection between the Weisfeiler-Leman algorithm and a standard
gradient descent algorithm, the Franke-Wolfe algorithm
\cite{KerstingMGG14}.

Declarative distances tend to be easier to compute or at least
approximate. Graph kernels are usually designed to
be computed efficiently, the Weisfeiler-Leman kernels being a case in
point. Distances based on homomorphism densities or homomorphism
counting can at least be approximated by sampling. Computing them
exactly may be hard, though, even impossible. Just consider
homomorphism indistinguishability, that is, the problem of deciding
whether two graphs $G,H$ have distance $\delta_{\CF}(G,H)=0$ with
respect to some class $\CF$. The complexity of this problem varies
wildly from polynomial time to undecidable \cite{BokerCGR19}. For
example, for the class of all paths, the class of all cycles, the class of
all trees, or the classes of all graphs of tree width $k$, it is in
polynomial time \cite{Dvorak10,DellGR18} (also see
\cite{Seppelt24a}). For the class of all complete graphs, it is hard
for the complexity class C$_=$P \cite{BokerCGR19}, and for the class
of all planar graphs, it is undecidable \cite{MancinskaR20}. And of
course, for the class of all graphs, it is just the graph isomorphism
problem, which is in quasi-polynomial time \cite{Babai16}.

\section{Concluding Remarks}
In this paper, we discussed various approaches to measuring the
distance, or similarity, between graphs. Our focus was on the question
of how the different approaches relate to one another. Many interesting
questions remain open, both ``expressiveness'' questions about the
relation between different methods and algorithmic questions on
how to compute or approximate the metrics efficiently.

\printbibliography

\end{document}